\documentclass[reprint,nofootinbib,amsmath,amssymb,aps,prd,floatfix,longbibliography]{revtex4-2}
\usepackage{mathtools}
\usepackage[english]{babel}
\usepackage{amsthm}
\usepackage{graphicx}
\usepackage{dcolumn}
\usepackage{bm}
\usepackage{silence}
\usepackage{hyperref}
\usepackage{microtype}

\makeatletter
\let\REVTeX@makecaption\@makecaption
\usepackage{subcaption}
\let\@makecaption\REVTeX@makecaption
\makeatother

\newtheorem{definition}{Definition}
\newtheorem{theorem}{Theorem}[section]
\newtheorem{proposition}{Proposition}

\newcommand\kink{\operatorname{kink}}
\newcommand\M{\mathrm{M}}
\newcommand\const{\text{const}}

\renewcommand\d{\mathrm{d}}
\renewcommand\i{\operatorname{i}}

\begin{document}

\preprint{}

\title{Wormhole Nucleation via Topological Surgery in Lorentzian Geometry}

\author{Alessandro Pisana}
\email{apisana@brocku.ca}
\affiliation{Department of Physics\\\mbox{Brock University, 1812 Sir Isaac Brock Way, St.~Catharines, Ontario, L2S 3A1, Canada}}

\author{Barak Shoshany}
\email{bshoshany@brocku.ca}
\homepage{https://baraksh.com/}
\thanks{corresponding author}
\affiliation{Department of Physics\\\mbox{Brock University, 1812 Sir Isaac Brock Way, St.~Catharines, Ontario, L2S 3A1, Canada}}

\author{Stathis Antoniou}
\email{stathis.antoniou@gmail.com}
\affiliation{School of Applied Mathematical and Physical Sciences\\National Technical University of Athens, Zografou Campus, GR-157 80, Athens, Greece}

\author{Louis H. Kauffman}
\email{loukau@gmail.com}
\affiliation{Department of Mathematics, Statistics and Computer Science\\University of Illinois at Chicago, 851 South Morgan Street, Chicago, IL, 60607-7045}
\affiliation{}
\affiliation{International Institute for Sustainability with Knotted Chiral Meta Matter (WPI-SKCM2)\\Hiroshima University, 1-3-1 Kagamiyama, Higashi-Hiroshima, Hiroshima 739-8526, Japan}

\author{Sofia Lambropoulou}
\email{sofia@math.ntua.gr}
\affiliation{School of Applied Mathematical and Physical Sciences\\National Technical University of Athens, Zografou Campus, GR-157 80, Athens, Greece }

\begin{abstract}
We construct a model for the nucleation of a wormhole within a Lorentzian spacetime by employing techniques from topological surgery and Morse theory. In our framework, a 0-surgery process describes the neighborhood of the nucleation point inside a compact region of spacetime, yielding a singular Lorentzian cobordism that connects two spacelike regions with different topologies. To avoid the singularity at the critical point of the Morse function, we employ the Misner trick of taking a connected sum with a closed 4-manifold---namely $\mathbb{CP}^{2}$---to obtain an everywhere nondegenerate Lorentzian metric. This connected sum replaces the naked singularity with a region containing closed timelike curves. The obtained spacetime is nonsingular, but violates all the standard energy conditions. Our construction, thus, shows that a wormhole can be ``created'' without singularities in classical general relativity.
\end{abstract}

\maketitle

\section{Introduction}

\subsection{Motivation}

The roots of the idea that non-linearities in Einstein's field equations could give rise to topologically nontrivial configurations of spacetime trace back to Wheeler and Misner \cite{wheeler, MISNER1957525}, who introduced the \emph{geon}, an isolated gravitational-electromagnetic entity such that classical mass and unquantized electric charge can be recovered from its non-trivial topological features---obtaining \emph{``mass without mass''} and \emph{``charge without charge''}.

Among the simplest examples of geons are wormholes: handles attached to spacetime. These have been extensively investigated as theoretical solutions of Einstein field equations \cite{VisserBook, LoboBook, ShoshanyFTLTT}. However, all known classical Lorentzian wormhole solutions are eternal, since their nucleation or annihilation is usually relegated to quantum gravity effects \cite{HAWKING1978349}. Our work seeks to push the boundaries of the classical theory and explore whether these topologically non-trivial configurations can arise from trivial ones in purely classical settings.

Important constraints on this investigation are imposed by two theorems due to Geroch \cite{gerochtopology}. The first theorem states that it is always possible to interpolate between two closed $3$-manifolds via a Lorentzian cobordism. However, if the two $3$-manifolds are non-homeomorphic, Geroch's second theorem requires that the spacetime contains either singularities or closed timelike curves (CTCs).

Two fundamental perspectives emerge in addressing classical topology changes. The first prioritizes causality and leads to scenarios where the interpolating spacetimes include regions with metric degeneracies \cite{horowitz1991topology, Horowitz_1991}, Morse singularities \cite{Louko_1997, Dowker_2000, dowker2002topology}, or instantons \cite{10.1143/PTP.65.1443, HOSOYA1989117, REY1990146, PhysRevD.50.2662, PhysRevD.90.124015, battara}. The second emphasizes the equivalence principle and permits topology changes at the expense of causality \cite{BordeImpossible, borde1994topology}, leading to conceptual issues such as causality paradoxes \cite{Shoshany_Hauser_2020, Shoshany_Wogan_2021, ShoshanyStober}.

In this work we adopt the latter perspective, maintaining an agnostic viewpoint on whatever chronology and energy conditions should be imposed on the spacetime---especially since CTCs may well be an inevitable byproduct of the wormholes themselves \cite{wormholestimemachinesWEC}, as well as other forms of superluminal travel \cite{ShoshanySnodgrass}.

Specifically, we implement topological surgery techniques \cite{antoniou2018topological, REINHART1963173} and results from Morse theory \cite{kostantinov, MYuKonstantinov_1986, HFDowker_1998} for constructing an explicit Lorentzian cobordism, obtaining an everywhere non-degenerate Lorentzian metric in the neighborhood of the wormhole nucleation point, addressing the challenges of regularizing the gradients of Morse functions \cite{Yodzis1972}. The conclusions we present, apart from the investigations of the energy conditions on these spacetimes, are purely kinematic and thus applicable not only to general relativity but also to its various modifications.

Our approach follows the more general framework of Hawking and Borde \cite{chronologyprotection, borde1994topology, BordeImpossible}, which relaxes the closure condition for the initial and final $3$-manifolds. In this setup, topological transitions are physically constrained to occur within a compact timelike cylinder, where the transition region has a boundary that can vary from timelike to spacelike.

The number of times this boundary changes its causal character constitutes additional boundary data that can be encoded as a topological invariant called the \emph{kink number} \cite{FINKELSTEIN1959230, kinkextensions}. This invariant is crucial in determining whether a spin structure exists on the spacetime, and, in turn, gives rise to selection rules governing the possible topology changing manifolds \cite{selectionrules}. Notably, these rules do not appear to prohibit the creation of wormholes, as long as they are created in pairs \cite{Gibbons_1993, gibbons2011topology}.

We do not address the difficulties that compactly generated Cauchy horizons pose for quantum field theories---most notably the divergent renormalized stress‑energy tensor that signals explosive particle creation \cite{Anderson1986, QFTWald, PhysRevD.58.023501}. Such issues, though important, lie beyond the scope of the present classical analysis.

\subsection{Outline}

In Section~\ref{sec:topologicalsurgery}, we summarize the key techniques and results of topological surgery via Morse functions. These techniques allow us to select the manifold on which the Lorentzian metric describing the nucleating wormhole is built.

In Section~\ref{sec:topologychange} we review topology changes in classical general relativity, focusing on Geroch's theorems and their extensions to our framework where the transitions occur within a compact spacetime region. Within this setup we analyze the topological obstructions to the existence of a Lorentzian structure arising from the kink number and the Poincaré–Hopf theorem. We then show the necessity of choosing $\mathbb{CP}^{2}$ as the closed $4$‑manifold that desingularizes our otherwise singular Lorentzian cobordism.

In Section~\ref{sec:wormholenucleation} we present the central result of this work: the construction of an explicit, regular Lorentzian metric in the neighborhood of the transition point. The construction requires singular Lorentzian metrics on both $\mathbb{CP}^{2}$ and the spacetime undergoing the topology change, and we discuss the resulting violations of the energy and chronology conditions in these two singular spacetimes.

In Section~\ref{sec:gluing} we employ the \emph{rigging} formalism to perform the connected sum of the two singular Lorentzian manifolds, concentrating on the smooth gluing along hypersurfaces whose causal character varies pointwise. We show that this procedure does not introduce thin shells.

Finally, Section~\ref{sec:conclusions} outlines potential applications of CTCs to the desingularization of singular spacetimes and suggests directions for future work.

Throughout this paper, we utilize Wolfram Mathematica with the OGRe package \cite{Shoshany2021_OGRe}, Python with the OGRePy package \cite{Shoshany2025_OGRePy}, and Maple for both symbolic and numerical calculations.

\section{Topological surgery}
\label{sec:topologicalsurgery}

Topological surgery is the natural framework to study topology changes of $m$-dimensional manifolds. Here and in the rest of this work we take $m = 3$. Using topological surgery allows us to identify the $4$-dimensional manifold on which we aim to define the Lorentzian metric describing the nucleation of a wormhole, as well as any topological obstructions to the existence of this metric.

Given that our analysis is restricted to $4$-dimensional spacetimes, we are considering only $3$-dimensional $n$-surgeries. For clarity, we adopt the notation and definitions presented in Ref.~\cite{antoniou2018topological}. Further examples, additional details and proofs of the statements in this section can be found in Refs.~\cite{Topological_Surgery_Small_Large, antoniou2018blackholestopologicalsurgery, detectingvisualizing3dimensionalsurgery, Milnor1965, MilnorStiefel, MorseMilnor}.

Given two closed (compact without boundaries) $3$-manifolds $\Sigma_i$ and $\Sigma_f$, we define the topological $3$-dimensional $n$-surgery as follows:

\begin{definition}
    A $3$-dimensional $n$-surgery, where $n \in \{0,1,2\}$,  is the topological process of creating a new $3$-manifold $\Sigma_f$ out of a given $3$-manifold $\Sigma_i$ by removing a framed $n$-embedding $h: S^{n} \times D^{3-n}\hookrightarrow \Sigma_i$, and replacing it with $D^{n+1} \times S^{3-n-1}$, using the gluing homeomorphism $h$ along the common boundary $S^{n} \times S^{2-n}$:
    \[
        \Sigma_f = \overline{\Sigma_i \setminus h(S^n \times D^{3-n})} \cup_{h|_{S^n \times S^{2-n}}} (D^{n+1} \times S^{2-n}).
    \]
    \label{def:surgery}
\end{definition}

Here and in the rest of this paper, $S^n$ represent topological spheres and $D^n$ represent closed disks.

It is easily verified that the surgery process does not alter the Euler characteristic $\chi$ of $\Sigma_i$, so that $\chi(\Sigma_i) = \chi(\Sigma_f) = 0$ as it should be. This is true even if $\Sigma_i$ and $\Sigma_f$ have boundaries; in this case, which is the one of interest, it is sufficient to assume that the surgery takes place in the interior of $\Sigma_i$ and that therefore $\partial \Sigma_i \cong \partial \Sigma_f$.

Following Definition~\ref{def:surgery}, there are three kinds of $3$-dimensional surgeries---or, adopting the notation of Ref.~\cite{Yodzis1972} with $n = \lambda -1$, three spherical modifications of type $\lambda -1$:

\begin{itemize}
    \item \textbf{$3$-dimensional $0$-surgery}: An embedding of $S^0 \times D^3$ is replaced by $D^1 \times S^2$. Both embeddings have the same boundary, and the surgery consists of pasting one in place of the other along the common boundary. This surgery corresponds to the creation of a wormhole, which is the focus of this work. In the rest of this section we will thus specialize in this type of surgery, although appropriate generalizations for the other types of surgeries can be found in Refs.~\cite{antoniou2018topological, Topological_Surgery_Small_Large, antoniou2018blackholestopologicalsurgery, Antoniou2017-mt}.
    \item \textbf{$3$-dimensional $1$-surgery}: An embedding of a solid torus $S^1 \times D^2$ is replaced by another solid torus $D^2 \times S^1$. Although this surgery does not seem to have a direct physical interpretation, it plays a central role in Theorem~\ref{th:knotsurgery} below.
    \item \textbf{$3$-dimensional $2$-surgery}: The reverse process of the $0$-surgery, physically representing the annihilation of a wormhole.
\end{itemize}

The relevance of $3$-dimensional surgery is that any $3$-manifold can be obtained from any other by a finite number of elementary operations. This result can be stated as follows:

\begin{theorem}
    Every closed, connected, orientable $3$-manifold can be obtained by surgery on a knot or a link on $S^3$.
    \label{th:knotsurgery}
\end{theorem}

In this theorem, by Wallace \cite{Wallace_1960} and Lickorish \cite{lickorish}, a \emph{knot} is an embedding of $S^1$ in $\mathbb{R}^3$ or $S^3$, and a \emph{link} is a collection of knots that do not intersect but can be linked together.

Two 3-dimensional manifolds obtained by surgery from each other are \emph{cobordant}, meaning there exist a \emph{cobordism} between them, defined as follows:

\begin{definition}
    A $4$-dimensional \textbf{cobordism} $(W; \Sigma_i, \Sigma_f)$ is a $4$-dimensional manifold $W$ whose boundary is the disjoint union of the closed $3$-manifolds $\Sigma_i, \Sigma_f: \partial W = \Sigma_i \sqcup \Sigma_f$. A $4$-dimensional cobordism $(W; \Sigma_i, \Sigma_f)$ is an \textbf{h-cobordism} if the inclusion maps $\Sigma_i \hookrightarrow W$ and $\Sigma_f \hookrightarrow W$ are homotopy equivalences.
\end{definition}

Several definitions of cobordisms appear in the literature; we choose this one for convenience and consistency of notation. It is also possible to add extra information, as in the case of Ref.~\cite{REINHART1963173}, where a vector field with specific boundary conditions on the two $3$-manifolds is considered.

Given an initial $3$-manifold $\Sigma_i$, the cobordism between $\Sigma_i$ and $\Sigma_f$, obtained from $\Sigma_i$ by a $0$-surgery, is the trace of the $0$-surgery:

\begin{definition}
    The \textbf{trace} of the $0$-surgery removing $S^{0} \times D^{3} \subset \Sigma_i$ is the cobordism $(W; \Sigma_i, \Sigma_f)$ obtained by attaching the $4$-dimensional $1$-handle $D^{1} \times D^{3}$ to $\Sigma_i\times I$ at $S^{0} \times D^{3} \times \{ 1 \} \subset \Sigma \times \{ 1 \} $. The $4$-dimensional $1$-handle is the disk $D^4$ considered in its homeomorphic form $D^1 \times D^3$ which has boundary $\partial (D^{1} \times D^{3}) = (S^{0} \times D^{3}) \sqcup (D^{1} \times S^{2})$.
\end{definition}

Here, $I=[0,1]$ is the unit interval. On the one hand, this surgical process can be viewed as the continuous transition, within the handle $D^{1} \times D^{3}$, from the boundary component $S^0 \times D^3$ to its complement $D^1 \times S^2$. On the other hand, from a local perspective, the surgery can be described using \emph{Morse functions}. Let us first define the notion of regular and critical points:

\begin{definition}
    Let $f: M \rightarrow N$ be a differentiable map between two manifolds $M$ and $N$ of dimensions $m$ and $n$ respectively.
    \begin{enumerate}
        \item A \textbf{regular point} of $f$ is a point $x \in M$ where the differential $\d f(x) : \mathbb{R}^m \rightarrow \mathbb{R}^n$ is a linear map of maximal rank, that is, $\operatorname{rank}(\d f(x))=$ $\min (n, m)$.
        \item A \textbf{critical point} of $f$ is a point $x \in M$ which is not regular.
        \item A \textbf{regular value} of $f$ is a point $y \in N$ such that every $x$ in $f^{-1}(\{y\}) \subseteq$ $M$ (the \textbf{level set} of $f$ at $y$) is regular, including the empty case $f^{-1}(\{y\})=\emptyset$.
        \item A \textbf{critical value} of $f$ is a point $y \in N$ which is not regular.
    \end{enumerate}
\end{definition}

Now we can define Morse functions:

\begin{definition}
    Let $f: M \rightarrow \mathbb{R}$ be a differentiable function on an $m$-dimensional manifold.
    \begin{enumerate}
        \item A critical point $x \in M$ of $f$ is \textbf{non-degenerate} if the Hessian matrix $H(x)=\left(\frac{\partial^2 f}{\partial x_i \partial x_j}\right)$ is invertible.
        \item The index $\operatorname{Ind}(x)$ of a non-degenerate critical point $x \in M$ is the number of negative eigenvalues in $H(x)$, so that with respect to appropriate local coordinates the quadratic term Q in the Taylor series of $f$ near $x$ is given by
        \[
            Q = -\sum_{i=1}^{\operatorname{Ind}(x)}\left(h_i\right)^2+\sum_{i=\operatorname{Ind}(x)+1}^m\left(h_i\right)^2.
        \]
        \item The function $f$ is \textbf{Morse} if it has only non-degenerate critical points.
    \end{enumerate}
\end{definition}

A distinctive feature of Morse functions is that topological transitions are described by means of their level sets. Two level sets of a Morse function are diffeomorphic if there are no critical level sets between them. Consequently, for a transition to occur, the level sets must intersect a critical point.

\begin{figure}
    \centering
    \includegraphics[width=1\linewidth]{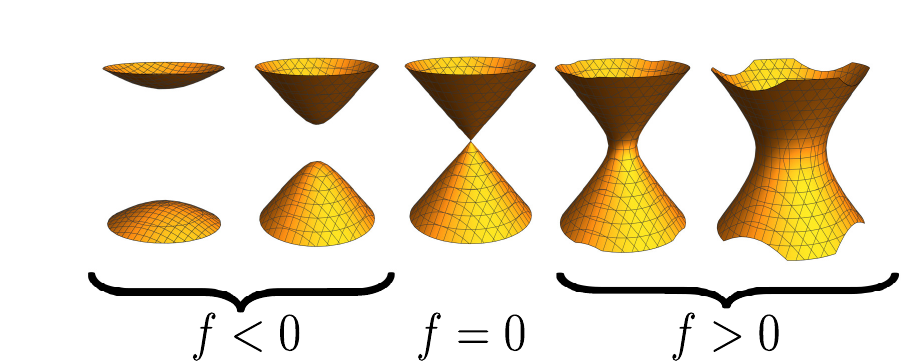}
    \caption{The $3$-dimensional $0$-surgery removes an embedding of $S^0 \times D^3$ and replaces it with $D^1 \times S^2$. This process, around the critical point, can be represented by a Morse function such as $f = -x^2 + y^2 + z^2$ (suppressing one dimension), whose level sets are shown above.}
    \label{fig:figure1}
\end{figure}

In Fig.~\ref{fig:figure1}, the level sets  $f<0$ are diffeomorphic to each other, as are the level sets $f>0$. However, between them there is the critical level set: the double cone $f = 0$.

The bridge between the two descriptions---namely, that via Morse functions and that via cobordisms---is provided by the following proposition:

\begin{proposition}
    Let $f: W \rightarrow I$, where $I$ is the unit interval, be a Morse function on a $4$-dimensional cobordism ($W; \Sigma_i, \Sigma_f$) between manifolds $\Sigma_i$ and $\Sigma_f$ with
    \[
        f^{-1}(\{0\})=\Sigma_i, \quad f^{-1}(\{1\})=\Sigma_f
    \]
    and such that all critical points of $f$ are in the interior of $W$.
    \begin{enumerate}
        \item If $f$ has no critical points then $\left(W; \Sigma_i, \Sigma_f\right)$ is a trivial $h$-cobordism, with a diffeomorphism
        \[
            \left(W; \Sigma_i, \Sigma_f\right) \cong \Sigma_i \times(I; \{0\}, \{1\})
        \]
        which is the identity on $\Sigma_i$.
        \item If $f$ has a single critical point of index $i$ then $W$ is obtained from $\Sigma_i \times I$ by attaching a $4$-dimensional $i$-handle using an embedding $S^{i-1} \times D^{4-i} \hookrightarrow \Sigma_i \times\{1\}$, and $\left(W; \Sigma_i, \Sigma_f\right)$ is an elementary cobordism of index $i$ with a diffeomorphism
        \[
            (W; \Sigma_i, \Sigma_f) \cong\left((\Sigma_i \times I) \cup (D^i \times D^{4-i}); \Sigma_i \times\{0\}, \Sigma_f\right).
        \]
    \end{enumerate}
\end{proposition}

A fundamental property of any Morse function is that, in a small enough neighborhood of a critical point $p$, there always exists a coordinate system $x^i$ such that the Morse function for a $0$-surgery takes the form:
\begin{equation}
    f(x) = f(p) - (x^0)^2 + (x^1)^2 + (x^2)^2 + (x^3)^2,
    \label{eq:localmorse}
\end{equation}
where, without loss of generality, we can set $f(p) = 0$.

We conclude this section with the following theorem:

\begin{theorem}
    Every $4$-dimensional manifold $W$ admits a Morse function $f: W \rightarrow \mathbb{R}$.
    \label{th:existmorse}
\end{theorem}

This last theorem guarantees that, in a neighborhood of the topological transition point, it is always possible to construct a singular Lorentzian metric. In fact, on the $W$ cobordism describing the formation of a wormhole, since there always exists a Morse function, it is always possible to find a coordinate system such that in the neighborhood of the critical point, the Morse function takes the form (\ref{eq:localmorse}). From the gradient of this function, together with an appropriately chosen Riemannian metric, it  is then possible to obtain the singular Lorentzian metric, as shown in Section~\ref{sec:wormholenucleation}. This gradient, normalized with respect to the Riemannian metric, can then be considered as a map from $S^3 \rightarrow S^3$, whose degree must be $\pm 1$.

\section{Topology changes in classical general relativity}
\label{sec:topologychange}

\subsection{Preliminaries}

As explained in the previous section, our focus is on changes in the spatial topology of spacelike submanifolds of a $4$-dimensional spacetime within the framework of classical general relativity. Specifically, we aim to investigate the formation of a wormhole in a spacetime that admits an everywhere non-degenerate Lorentzian metric, even if this entails the presence of CTCs---which may well be harmless if confined to microscopic regions of spacetime \cite{topologychangeandmonopole}.

Where not otherwise specified, we define a spacetime as follows:

\begin{definition}
    A spacetime $(W, g^L_{\mu\nu})$ is a $4$-dimensional manifold $W$, possibly compact, together with an everywhere non-degenerate time-orientable Lorentzian metric $g^L_{\mu\nu}$.
\end{definition}

The signature convention we adopt is $(-,+,+,+)$. If the spacetime is not time-orientable, it is always possible to consider a time-orientable double covering, with minor adjustments to the results presented in this work \cite{Sorkin1986, gerochtopology}. Notably, it is possible to neglect time-orientability and achieve topology changes, as in \cite{Friedman}.

Our study of topological transitions begins with the following theorem, due to Geroch and attributed to Misner \cite{gerochtopology}:

\begin{theorem}
    Let $\Sigma_i$ and $\Sigma_f$ be two closed $3$-manifolds. Then there exists a compact spacetime $W$ whose boundary is the disjoint union of $\Sigma_i$ and $\Sigma_f$, and in which $\Sigma_i$ and $\Sigma_f$ are both spacelike.
\end{theorem}

As reported by Geroch, two closed 3-manifolds are always cobordant \cite{gerochtopology}. In fact, two closed manifolds are cobordant if and only if their Stiefel-Whitney numbers are the same \cite{MilnorStiefel}. If, in addition, the manifolds are oriented, then their Pontryagin numbers must also be equal. However, all closed $3$-manifolds have the same Pontryagin and Euler numbers (in particular, Pontryagin numbers for odd-dimensional manifolds are always zero). For more details on the subject, please see Refs.~\cite{MilnorStiefel, Thom1954, Stong}.

The main challenge is then the requirement of an everywhere non-degenerate, time-orientable Lorentzian metric on $W$. Such a metric exists on the cobordism $W$ if and only if an everywhere non-singular vector field exists on $W$. A formal proof is provided by Steenrod \cite{Steenrod}, but the underlying idea is simple.

A Lorentzian metric $g^{L}_{\mu\nu}$ can always be obtained from a Riemannian one $g^{R}_{\mu\nu}$---which only requires paracompactness---together with an everywhere non-singular vector field $V^\mu$ \cite{Hawking_Ellis_1973, topologicaltecniques}:

\begin{equation}
   g^{L}_{\mu\nu} = g^{R}_{\mu\nu} - 2 \, \frac{g^{R}_{\mu\alpha}V^{\alpha}g^{R}_{\nu\beta}V^{\beta}}{g^{R}_{\rho\sigma}V^{\rho}V^{\sigma}}.
    \label{eq:lorentzianfromriemannian}
\end{equation}

The index on $V^\mu$ is lowered (and its norm is computed) with $g^{R}_{\mu\nu}$. The converse is also true \cite{topologicaltecniques}: given a Lorentzian metric $g^L_{\mu\nu}$, one can diagonalize it with respect to a chosen Riemannian metric $g^R_{\mu\nu}$. The vector field $V^\mu$ then corresponds to the unique future-pointing unit eigenvector associated with the negative eigenvalue $\lambda$, satisfying $(g^L_{\mu\nu} - \lambda g^R_{\mu\nu}) V^{\mu} =0$ and $g^R_{\mu\nu} V^\mu V^\nu = 1$.

The vector field $V^\mu$, by construction, is timelike with respect to the Lorentzian metric~\eqref{eq:lorentzianfromriemannian}. The existence of this non-singular $V^\mu$ on $W$ depends on the boundary conditions imposed on $\Sigma_i$ and $\Sigma_f$: if $V^\mu$ is everywhere incoming or outgoing normal on $\Sigma_i$ and $\Sigma_f$, so that both are spacelike, then---by the Poincaré-Hopf theorem---the Euler characteristic of $W$ must vanish.\footnote{The Poincaré-Hopf theorem states that $\sum i_v = \chi(W)$, where $i_v$ is the index of an isolated singularity of the vector field. Although $\chi(W) = 0$ implies that the sum of the indices is zero and thus, \emph{a priori}, there may be several singularities of opposite sign, an important result of Milnor \cite{Milnor1965} guarantees that if $\chi(W) = 0$ there exists a vector field that is everywhere non-singular on $W$.}. The same conclusion is reached if $V^\mu$ is normal entering on $\Sigma_i$ and normal exiting on $\Sigma_f$ \cite{REINHART1963173}.

We adopt the convention that the initial manifold $\Sigma_i$ is the one where the vector field enters, while $\Sigma_f$ is the one where it exits; we denote the transition by $\Sigma_i \rightarrow \Sigma_f$. When the vector field is everywhere entering (exiting) on $\partial W$, we then write $\Sigma_i \sqcup \Sigma_f \rightarrow \emptyset$ ($\emptyset \rightarrow \Sigma_i \sqcup \Sigma_f $).

The requirement $\chi(W)=0$, however, is not obstructive: if $\chi(W)\ne0$, a new cobordism $W'$ between $\Sigma_i$ and $\Sigma_f$ can be constructed via the Misner trick of taking the connected sum with a compact $4$-manifold $N$, $W' = W \# N$, chosen so that
\[
    \chi (W \# N) = \chi (W) + \chi(N) -2 = 0.
\]
For example, $N = S^{2} \times S^{2}$ increases $\chi$ by $2$, $N = S^1 \times S^3$ decreases $\chi$ by $2$, $N = \mathbb{CP}^{2}$ increases $\chi$ by $1$, and $N = \mathbb{RP}^{4}$ decreases $\chi$ by $1$.

Thus a Lorentzian cobordism between two closed spacelike $3$-manifolds always exists, provided the Euler characteristic is adjusted by an appropriate connected sum. This cobordism, as imposed by the second Geroch theorem \cite{gerochtopology}, must, however contain either singularities or CTCs if $\Sigma_i$ and $\Sigma_f$ are not diffeomorphic:

\begin{theorem}
    Let $W$ be a compact spacetime whose boundary is the disjoint union of two compact spacelike $3$-manifolds, $\Sigma_i$ and $\Sigma_f$. If $W$ has no CTCs, then $\Sigma_i$ and $\Sigma_f$ are diffeomorphic, and, furthermore, $W$ is topologically equivalent to $\Sigma_i \times [0,1]$.
    \label{secondgerochtheorem}
\end{theorem}

Given that on $W$ there exists an everywhere non-singular vector field, the interpretation of Theorem~\ref{secondgerochtheorem} is clear: if there are no singularities or CTCs, the integral curves of the vector field form a homeomorphism between $\Sigma_i$ and $\Sigma_f$, implying that they have the same topology. A change of topology can therefore occur only by admitting singularities (so that $W$ is non-compact) or by allowing CTCs that remain trapped in the interior of $W$ without reaching $\Sigma_f$.

In this work we consider a more realistic setup---studied by Hawking and Borde---in which the topological transition occurs between two non-compact spacelike hypersurfaces $S_i$ and $S_f$ and is confined within a timelike tube $W$ \cite{chronologyprotection, borde1994topology, BordeImpossible}. This tube intersects $S_i$ in a compact region $\Sigma_i$ and $S_f$ in a compact region $\Sigma_f$, as shown in Fig.~\ref{fig:figure2}.

\begin{figure}
    \centering
    \includegraphics[width=0.7\linewidth]{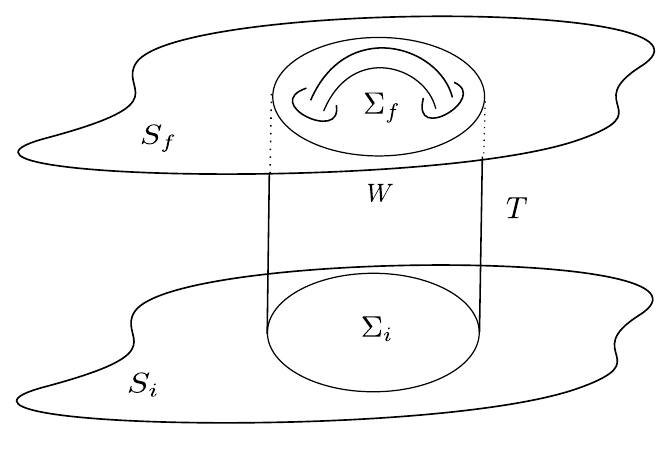}
    \caption{The transition occurs within a compact region $W$ that intersects the spacelike hypersurfaces $S_i$ and $S_f$ in $\Sigma_i$ and $\Sigma_f$, respectively. The cobordism therefore possesses an additional boundary component $T$.}
    \label{fig:figure2}
\end{figure}

Unlike the previous case, in which $\partial\Sigma_i$ and $\partial\Sigma_f$ are empty, we now have $\partial W = \Sigma_i \cup T \cup \Sigma_f$, where $T$ is an additional boundary component that intersects $S_i$ and $S_f$ in $\partial \Sigma_i$ and $\partial \Sigma_f$, respectively, so that $\partial \Sigma_i \cong \partial \Sigma_f$ and $T$ is topologically $\partial \Sigma_i \times [0,1]$. The specification of the vector field on the boundary can always be arranged so that $\Sigma_i$ and $\Sigma_f$ are spacelike and $T$ is timelike \cite{borde1994topology}.

In this setup, the Poincaré-Hopf theorem is no longer valid and a generalization is required to take into account the changing causal character of $\partial W$. This generalization introduces the \emph{kink number}, the topic of the next subsection. As in the closed case of Theorem~\ref{secondgerochtheorem}, we emphasize that $W$ must contain singularities or CTCs for a change of topology to occur, as proved by the following theorem due to Hawking \cite{chronologyprotection}:

\begin{theorem}
    If there is a timelike tube $W$ connecting surfaces $S_i$ and $S_f$ of different topology, then $W$ contains CTCs.
    \label{th:hawking}
\end{theorem}

Before proceeding with the definition and calculation of the kink number, we point out some consequences that this compact setup has on the energy conditions.

Throughout this work, we shall refer to the following as the \emph{standard energy conditions} of classical general relativity (with no cosmological constant) \cite{ShoshanyFTLTT, Curiel2017}, where $T_{\mu\nu}$ is the stress-energy tensor:
\begin{itemize}
    \item \textbf{Strong energy condition (SEC)}
    \[
        \left( T_{\mu\nu} -\tfrac{1}{2}Tg^L_{\mu\nu} \right)t^\mu t^\nu \ge 0, \qquad \forall\,t^\mu: g^{L}_{\mu\nu}t^\mu t^\nu < 0.
    \]
    \item \textbf{Null energy condition (NEC)}
    \[
        T_{\mu\nu} k^\mu k^\nu \ge 0, \qquad \forall\,k^\mu: g^{L}_{\mu\nu}k^\mu k^\nu = 0.
    \]
    \item \textbf{Weak energy condition (WEC)}
    \[
        T_{\mu\nu} t^\mu t^\nu \ge 0, \qquad \forall\,t^\mu: g^{L}_{\mu\nu}t^\mu t^\nu < 0.
    \]
    \item \textbf{Dominant energy condition (DEC)}
    \[
        \begin{aligned}
            &T_{\mu\nu} t^\mu t^\nu \ge 0, \\[1mm] &
            g^L_{\mu\nu}T^\mu_{\alpha}T^\nu_{\beta}t^\alpha t^\beta \le 0,
        \end{aligned}
        \qquad \forall\,t^\mu: g^{L}_{\mu\nu}t^\mu t^\nu < 0.
    \]
\end{itemize}

None of these conditions is independent of the others, due to a hierarchy among them. If the NEC is violated, then all the other energy conditions are violated as well. If the WEC is violated, then the DEC is also violated. Given this hierarchy, in this work we concentrate only on violations of the NEC and WEC.

As discussed in Ref.~\cite{Hawking_Ellis_1973}, the presence of CTCs and the consequent violations of causality may be physically more troubling than the mere lack of predictability stemming from singularities. Indeed, a theorem due to Hawking demonstrates that the formation of singularities is driven by energy conditions rather than by chronological assumptions, which are usually invoked in other singularity theorems\footnote{See Theorem~4, p.~272 in Ref.~\cite{Hawking_Ellis_1973}}. This result is strengthened by a series of theorems due to Tipler \cite{TIPLER19771,Tipler1978} and Lee \cite{Lee}.

In particular, the formation of singularities in topology-changing spacetimes---even in our compact-transition setup---appears unavoidable if the energy conditions hold\footnote{See Theorem~4, p.~20 in Ref.~\cite{TIPLER19771}}. Moreover, even the mere existence of a compact, non-simply-connected region in spacetime forces violations of the energy conditions, as demonstrated by a theorem due to Gannon \cite{Gannon}. Because classical matter generally satisfies the energy conditions, it therefore seems unlikely that a wormhole can nucleate under an evolution which is nearly classical \cite{Topologicalcensorship}.

We show in Section~\ref{sec:wormholenucleation} that our results agree with these theorems: we indeed find regions where all the standard energy conditions are violated.

\subsection{The kink number}
\label{subsec:kinknumber}

The kink number is a topological invariant first introduced by Finkelstein and Misner during the heyday of geometrodynamics, under the name ``M-geon'' \cite{FINKELSTEIN1959230, CHAMBLIN1994357, kinksandtopologychange, Low}. Roughly speaking, it is an integer that measures the number of times the light-cone tips around a chosen submanifold $\Sigma$ of a spacetime $(W, g^{L}_{\mu\nu})$ and classifies Lorentzian metrics up to homotopy \cite{CHAMBLIN1994357, chamblin1996topology}. Several definitions exist in the literature, but they all lead to equivalent results; here we adopt the one in Ref.~\cite{kinksandtopologychange}.

Specifically, let $(W, g^{L}_{\mu\nu})$ be a spacetime and $\Sigma \subset W$ a connected, orientable $3$-submanifold. If $W$ is compact, $\Sigma$ may be $\partial W$. Consider the unit-sphere bundle over $W$, $S(W) \rightarrow W$, whose fiber $S_p(W) \hookrightarrow S(W)$ is the $3$-sphere of directions at $p \in W$. When restricted to $\Sigma$, we may consistently choose a unit normal $N^{\mu}$ of $\Sigma$, which is a global section $N^{\mu} \in \Gamma(W, S(\Sigma))$ of the $6$-dimensional bundle $S(\Sigma) \rightarrow \Sigma$.

Given the Lorentzian metric $g^{L}_{\mu\nu}$, we choose an auxiliary Riemannian metric $ g^{R}_{\mu\nu}$ and solve the eigenvalue problem $g^{L}_{\mu\nu}V^{\nu} = \lambda g^{R}_{\mu\nu}V^{\nu}$ such that the relation~\eqref{eq:lorentzianfromriemannian} is satisfied. There is a unique eigenvector with negative eigenvalue, normalized so that $g^{R}_{\mu\nu}V^{\mu}V^{\nu} = 1$. Its restriction to $\Sigma$ yields another global section $V^{\mu} \in \Gamma(W, S(\Sigma))$.

The kink number of $V^{\mu}$ with respect to $\Sigma$ is the signed intersection number of the sections $V^\mu$ and $N^\mu$:
\[
    \kink(\Sigma; V^{\mu}) \equiv \sum_{i}^{m} \mathrm{sgn}(\chi_i),
\]
where $\chi_i$ are the intersection points and $\mathrm{sgn}(\chi_i) = +1$ iff the orientation of $S(\Sigma)$ at $\chi_i$ equals the product of the orientations of $V^{\mu}$ and $N^{\mu}$, and $-1$ otherwise.

Because $\Sigma$ is orientable and $3$-dimensional, one can choose a global frame $\{ (e^{\mu}_{0} = N^{\mu}, e^{\mu}_{i})~|~i = 1,2,3 \}$ on $\Sigma$, where $e^{\mu}_{0}$ is the inward unit normal and $e_{i}^\mu$ are tangent. Extending this frame to a collared neighborhood $C \cong \Sigma \times \left[0,1 \right] $, the decomposition
\[
    V^{\mu} = f^{0}e^{\mu}_{0} + f^{i}e^{\mu}_{i}, \qquad  (f^{0})^{2} + \sum_{i}(f^{i})^{2} = 1,
\]
defines a map $f^{a}: \Sigma \rightarrow S^{3}$, and
\[
    \kink(\Sigma; V^{\mu}) \equiv \deg(f^a).
\]

Since the space of Riemannian metrics on $W$ is contractible \cite{kinkstimemachines}, the kink number depends only on the Lorentzian metric:
\[
    \kink(\Sigma; V^{\mu}) \equiv \kink(\Sigma; g^{L}_{ab}).
\]
It does, however, change sign if the normal $N^\mu$ is reversed.

To compute this degree explicitly, following Refs.~\cite{harriott2006, harriott2005}, we introduce a flat metric on the collar $C$:
\begin{equation}
    \d s_{k}^{2} \equiv k_{\mu\nu}\d x^{\mu}\d x^{\nu} = \d w^{2} + \d s_{\Sigma}^{2},
    \label{eq:flatmetric}
\end{equation}
where $w\in\left[0,1 \right]$ and $\d s_{\Sigma}^{2}$ is a flat metric on $\Sigma$. Then
\begin{equation}
    \kink(\Sigma; g^{L}_{\mu\nu}) = \frac{1}{2\pi^2}\int \, J \, \d u^{1} \wedge \d u^2 \wedge \d u^3,
    \label{eq:kinkformula}
\end{equation}
in which $\{ u^{1}, u^2, u^3 \}$ are intrinsic coordinates on $\Sigma$ and
\begin{equation}
    J \equiv \det\begin{pmatrix}
        f^{0}      & f^{1}      & f^{2}      & f^{3}      \\
        D_{1}f^{0} & D_{1}f^{1} & D_{1}f^{2} & D_{1}f^{3} \\
        D_{2}f^{0} & D_{2}f^{1} & D_{2}f^{2} & D_{2}f^{3} \\
        D_{3}f^{0} & D_{3}f^{1}   & D_{3}f^{2}     & D_{3}f^{3}
    \end{pmatrix},
    \label{eq:jacobian}
\end{equation}
where $D_{\mu}f^{a} \equiv \partial_{\mu}f^{a} + \omega_{\mu b}^{a} f^{b}$ and $\omega_{\mu b}^{a}$ are the connection coefficients of $k_{\mu\nu}$.

Assuming a compact manifold $W$ with non-empty boundary such that $\partial W = \Sigma$, taking the inward normal as positive, the generalized Poincaré-Hopf theorem then reads \cite{kinksandtopologychange, building, Low, chamblin1996topology, CHAMBLIN1994357}
\begin{equation}
    \sum {i_{v}} = \chi(W)-\kink(\partial W;V^{\mu}),
    \label{eq:generalizedtheorem}
\end{equation}
where the sum is over the indices of isolated singularities of $V^\mu$. In the more general case of $\partial W = \Sigma_0 \sqcup \Sigma_1 \sqcup \cdots \sqcup \Sigma_m$,
\[
    \kink(\partial W;V^{\mu}) = \sum_{j = 0}^{m}~\kink(\Sigma_j;V^{\mu}).
\]

An everywhere non-singular Lorentzian metric on $W$ exists iff
\[
    \chi(W) = \kink(\partial W;V^{\mu}).
\]

Finally, if the vector field is everywhere timelike or spacelike on $\partial W$ the kink number vanishes \cite{building}, reducing to the case studied by Geroch and Reinhart \cite{gerochtopology, REINHART1963173}.

\subsection{Trivial 2+1 cobordisms}
\label{subsec:trivialcobordism}

As a starting point, we consider a trivial cobordism in which no change in topology occurs. On this cobordism we give a detailed calculation of the kink number and its application to the generalized Poincaré-Hopf theorem. This calculation will show that the dominant contribution to the kink number comes from the regions where the vector field changes from timelike to spacelike and vice versa---a result that we later exploit in the $4$-dimensional case.

In accordance with the setup presented in the previous section, and because the cobordism is trivial, we model the $3$-dimensional spacetime as a solid cylinder $W = D^1 \times D^2$, representing the transition $D^2 \rightarrow D^2$. The boundary $\partial W$ has three components: the timelike component $D^1 \times S^1 \cong [0,1] \times S^1$, the initial spacelike configuration $\{0\} \times D^2$, and the final spacelike configuration $\{1\} \times D^2$.

\begin{figure}
     \centering
     \begin{subfigure}{0.35\linewidth}
         \centering
         \includegraphics[width=\linewidth]{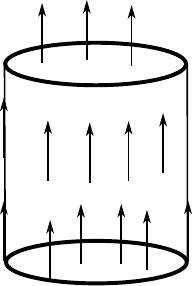}
         \caption{}
         \label{fig:figure3a}
     \end{subfigure}
     \hfill
     \begin{subfigure}{0.5\linewidth}
         \centering
         \includegraphics[width=\linewidth]{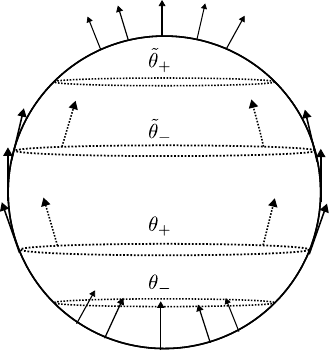}
         \caption{}
         \label{fig:figure3b}
     \end{subfigure}
    \caption{(a) The trivial $3$-dimensional cobordism $W$ is represented by a tube in $3$D Minkowski space. (b) To smooth out the corners, the solid cylinder $D^1 \times D^2$ is deformed to $D^3$, and the vector field is transformed accordingly.}
    \label{fig:figure3}
\end{figure}

For simplicity, we choose a constant vector field tangent to $[0,1] \times S^1 $, normal-entering on $\{ 0 \} \times D^2$ and normal-exiting on $\{ 1 \} \times  D^2$, as illustrated in Fig.~\ref{fig:figure3a}. This vector field extends trivially from the boundary to the interior of $W$ without zeros. In effect, we are selecting the timelike tube in $3$-dimensional Minkowski spacetime
\begin{equation}
    \d s^2 = -\d t^2 + \d r^2+r^2 \d \phi^2,
    \label{eq:Minkowskimetric}
\end{equation}
bounded by $0 \le t \le 1$ and $r \le 1$. On this tube the vector field $V^{\mu} = \partial_t$ arises from the diagonalization of~\eqref{eq:Minkowskimetric} with respect to the auxiliary Riemannian metric $\delta_{\mu\nu} = \mathrm{diag}(1,1,r^2)$.

Then, according to the generalized Poincaré-Hopf theorem,
\[
    \chi (W) = \kink\left(\partial W, V^{\mu} \right) = 1.
\]
Because the kink number vanishes on purely spacelike and timelike boundary components, the only contributions to the integral~\eqref{eq:kinkformula} come from the two circles $\{0\} \times S^1$ and $\{1\} \times S^1$, where $\partial W$ changes from timelike to spacelike and vice versa.

To compute explicitly, we smooth out the corners of the cylinder, deforming $D^{1} \times D^{2}$ to $D^{3}$ and the vector field accordingly---as in Fig.~\ref{fig:figure3b}, where we show a cross-section of the deformed tube together with the vector field. The final result is unaffected by this transformation, since the kink number is a topological invariant.

This smoothing process is done in such a way that, in spherical coordinates $0 \le \theta \le \pi$, $0 \le \phi \le 2\pi$---with $\theta$ measured from bottom to top \cite{kinkstimemachines}---the vector field is inward-normal in $0 \le \theta \le \theta_{-}$, outward-normal in $\tilde{\theta}_{+} \le \theta \le \pi$, and tangent in $\theta_{+} \le \theta \le \tilde{\theta}_{-}$. The two strips $[\theta_{-}, \theta_{+}]$ and $[\tilde{\theta}_{-}, \tilde{\theta}_{+}]$ are the regions where the field changes, respectively, from inward-normal to tangent, and from tangent to outward-normal. With $f^{\phi}=0$ everywhere on $\partial W$, we have
\[
    \begin{cases}
      \theta\in[0,\theta_{-}]\cup[\tilde{\theta}_{+},\pi]: & f^{t}\neq0,\ f^{\theta}=0,\\[2pt]
      \theta\in[\theta_{+},\tilde{\theta}_{-}]: & f^{t}=0,\ f^{\theta}\neq0,\\[2pt]
      \theta\in[\theta_{-},\theta_{+}]\cup[\tilde{\theta}_{-},\tilde{\theta}_{+}]: &
          f^{t}\neq0,\ f^{\theta}\neq0.
    \end{cases}
\]

The last step before computing the degree requires choosing the flat metric~\eqref{eq:flatmetric}. For $\partial W \cong S^2$, examples from Refs.~\cite{harriott2005,harriott2006} are
\begin{align*}
    \d s^{2}_{k,1} &= \d t^{2} + \cos^{2} \theta \, \d \theta^{2} + \sin^{2} \theta \, \d \phi^{2}, \\
    \d s^{2}_{k,2} &= \d t^{2} + \frac{a^{2}}{4}\sec^{4}\left(\frac{\theta}{2} \right)\, \d \theta^{2} + a^{2}\tan^{2} \left(\frac{\theta}{2} \right) \, \d \phi^{2},
\end{align*}
both yielding the same non-vanishing connection coefficients $\omega^{\theta}_{\phi\phi} = - \omega^{\phi}_{\phi \theta} = -1$. Hence
\[
    \begin{aligned}
      D_{\theta}f^{t} &=\partial_{\theta}f^{t}, & \qquad D_{\phi}f^{t} &=0,\\
      D_{\theta}f^{\theta} &=\partial_{\theta}f^{\theta}, & \qquad D_{\theta}f^{\phi} &=0,\\
      D_{\phi}f^{\theta} &=0, & \qquad D_{\phi}f^{\phi} &=f^{\theta},
    \end{aligned}
\]
which can now be inserted into~\eqref{eq:jacobian}. The determinant of the Jacobian is
\[
    J = f^{\theta} \Bigl( f^{t}\, \partial_{\theta}f^{\theta} - f^{\theta}\, \partial_{\theta}f^{t} \Bigr),
\]
from which the kink number---as defined in~\eqref{eq:kinkformula}, specialized to $3$-dimensions---is
\[
    \kink\left(\partial W, V^{\mu}\right) = \frac{1}{4\pi} \int_{S^{2}} J \, \d \theta \wedge \d \phi.
\]

Because the Jacobian vanishes outside the two transition strips, the only non-zero contribution is
\begin{equation}
    \kink\left(\partial W, V^{\mu}\right) =\frac{1}{2} \left( \int_{\theta_{-}}^{\theta_{+}} J \, \d \theta + \int_{\tilde{\theta}_{-}}^{\tilde{\theta}_{+}} J \, \d \theta \right).
    \label{eq:relevantintegral}
\end{equation}

We evaluate the integrals for the two regions separately. In the bottom strip the vector field changes from inward-normal to tangent. A possible choice for its components is
\[
    (f^{t}, f^{\theta}) = \frac{\left( 1-g(\theta), \, g(\theta) \right)}{\sqrt{1 - 2g(\theta) + 2g^2(\theta)}}.
\]
where $g(\theta)$ varies smoothly and monotonically from $0$ to $1$ as $\theta$ goes from $\theta_{-}$ to $\theta_{+}$, with $g'(\theta_{\pm}) = 0$. The square-root factor in the denominator ensures that the vector field has unit norm. Both the functional form of $g(\theta)$ and the interval $\Delta \theta \equiv \theta_{+} - \theta_{-}$ are arbitrary; the integral is independent of these choices. We then find
\begin{equation}
     \int_{\theta_{-}}^{\theta_{+}} J \, \d \theta
    = \int_{\theta_{-}}^{\theta_{+}} \frac{g(\theta)~g'(\theta)}{\left(1 - 2g(\theta) + 2g^2(\theta)\right)^{\frac{3}{2}}} \, \d \theta =1.
    \label{eq:bottomvectorfield}
\end{equation}

For the top strip, the calculation is analogous; here the vector field changes from tangent to outward-normal\footnote{Recall that, by convention, the outgoing normal is negative, so $f^{t}$ is negative in this region.}:
\[
    (f^{t}, f^{\theta}) = \frac{\left( -g(\theta) ,\, 1-g(\theta) \right)}{\sqrt{1 - 2g(\theta) + 2g^2(\theta)}},
\]
with $g(\theta)$ increasing from $0$ to $1$ on $[\tilde{\theta}_{-},\tilde{\theta}_{+}]$ and $g'(\tilde{\theta}_{\pm})=0$. Hence
\begin{equation}
    \int_{\tilde{\theta}_{-}}^{\tilde{\theta}_{+}} J \, \d \theta
    = \int_{\theta_{-}}^{\theta_{+}} J \, \d \theta = 1.
    \label{eq:topvectorfield}
\end{equation}

Substituting \eqref{eq:bottomvectorfield} and \eqref{eq:topvectorfield} into \eqref{eq:relevantintegral} gives
\[
    \kink\left(\partial W, V^{\mu}\right) = 1.
\]

This confirms that the vector field can be extended inside $W$ without singularities for the transition $D^2 \rightarrow D^2$. The same is not true, for example, for the transition $D^2 \sqcup D^2 \rightarrow \emptyset$. In the top region the field would have to change from tangent to inward-normal, and the integral~\eqref{eq:topvectorfield} would equal $-1$. Hence $\kink\left(\partial W, V^{\mu}\right) = 0$, so a singularity of index $+1$ must occur in $W$.

\subsection{The wormhole nucleation cobordism}
\label{subsec:wormholenucleation}

\begin{figure}
     \centering
     \begin{subfigure}{0.35\linewidth}
         \centering
         \includegraphics[width=\linewidth]{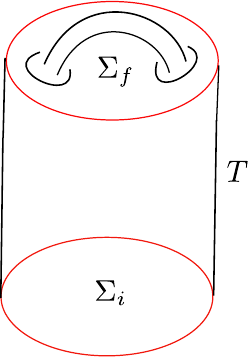}
         \caption{}
         \label{fig:figure4a}
     \end{subfigure}
     \hfill
     \begin{subfigure}{0.6\linewidth}
         \centering
         \includegraphics[width=\linewidth]{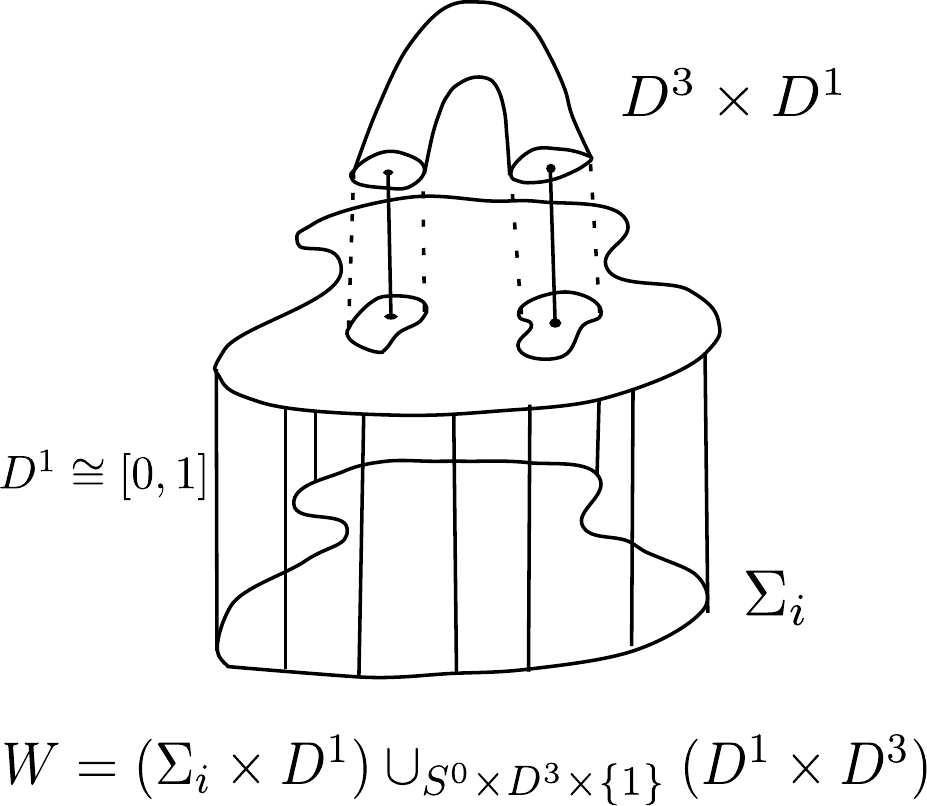}
         \caption{}
         \label{fig:figure4b}
     \end{subfigure}
    \caption{(a): The boundary $\partial W$ consists of three components---$\Sigma_i$, $\Sigma_f$, and $T$---whose interiors are spacelike, spacelike, and timelike, respectively. On each of these interiors the kink number is zero, independently of whether $\Sigma_i$ or $\Sigma_f$ is multiply connected. (b): The simplest cobordism for the $3$-dimensional $0$-surgery is obtained by attaching a ($D^1 \times D^3$)-handle to $\Sigma_i \times D^1$ along their common boundary $S^0 \times D^3$ at the point $\{1\} \in D^1$.}
    \label{fig:figure4}
\end{figure}

The result of the previous calculation can be generalized to any dimension: the kink number associated with the sphere
\[
    x_1^2 + x_2^2 + \cdots + x_n^2 = r_0^2
\]
in an $n$-dimensional Minkowski spacetime is equal to one. Moreover, as anticipated, the contribution to the kink number comes only from the regions of $\partial W$ where the vector field changes from timelike to spacelike and vice versa; it is therefore independent of the interior topology of $\Sigma_i$ and $\Sigma_f$.

It follows that, in the scenario we are investigating, where all the ``topological charge'' is confined within a compact region \cite{kinksandsingularities}, the spacetime $M$, which contains the compact cobordism $W$, may be assumed to be asymptotically flat. In terms of the boundary data on $\partial W$, represented schematically in Fig.~\ref{fig:figure4a}, this requirement translates as the statement that the vector field has the same number of kinks as a sphere in Minkowski space:
\[
    \kink(\partial W; g^L) = 1.
\]

From this assumed kink number we can check whether there are obstructions to the existence of a Lorentzian metric on $W$ arising from the generalized Poincaré-Hopf theorem~\eqref{eq:generalizedtheorem}. Topological surgery theory (see Section~\ref{sec:topologicalsurgery}) shows that a 3-dimensional $0$-surgery corresponds to attaching the handle $D^1 \times D^3$ to $\Sigma_i \times [0,1]$, as depicted in lower dimensions in Fig.~\ref{fig:figure4b}:
\[
    W = \Bigl( \Sigma_i \times [0,1] \Bigr) \cup_{S^0 \times D^3 \times \{1\}} \Bigl( D^1 \times D^3 \Bigr).
\]

Applying~\eqref{eq:generalizedtheorem} we obtain
\[
    \chi(W) = \chi(\Sigma_i) - 1 = \kink\left(\partial W, v \right) = 1.
\]

Because $\Sigma_i$ is a $3$-manifold with boundary, its Euler characteristic equals half that of its boundary $\partial \Sigma_{i}$. This boundary, in turn, is a closed oriented surface and thus its genus~$g$ satisfies $\chi(\partial \Sigma_{i}) = 2 - 2g$. Therefore:
\[
    \chi(\Sigma_i) = \frac{1}{2}\chi(\partial \Sigma_i) = 1 - g.
\]
Hence from $\chi(W) = 1$ we must have $g = -1$, which is impossible. Thus, the condition for an everywhere non-singular Lorentzian metric cannot be met for this simple cobordism; a singularity of index $-1$ is required\footnote{Since the initial manifold is assumed to be simply connected, we take $g = 0$. Thus, from the Poincaré-Hopf theorem, the index is $-1 - g = -1$. If the genus was positive, the index would be lower, but since the gradient of a Morse function can only have singularities of index $\pm 1$, that would correspond to multiple singularities of index $-1$ and therefore multiple successive surgeries.}---consistent with the index of the Morse gradient for a $0$-surgery, cf.~\eqref{eq:localmorse}.

Nevertheless, an everywhere non-degenerate Lorentzian metric can still be achieved via the Misner trick, by taking the connected sum with the complex projective plane $\mathbb{CP}^{2}$ to adjust the Euler characteristic. With the initial manifold $\Sigma_{i}=D^{3}$ the new cobordism is
\[
    W =\underbrace{\Big( \left( D^3 \times [0,1] \right) \cup_{S^0 \times D^3 \times \{1\}} \left( D^1 \times D^3 \right) \Big)}_{\mathbb{M}}\,\#\,\mathbb{CP}^{2}.
\]

The construction of this cobordism is the subject of Section~\ref{sec:wormholenucleation}. Before giving the explicit details, however, we discuss the selection rules for topology-changing spacetimes that arise from demanding a spin structure on $W$.

\subsection{Spin structure}
\label{subsec:spinstructure}

Among the problems that can plague Lorentzian cobordisms of topology-changing spacetimes is the possible obstruction to the existence of an $SL(2,\mathbb{C})$ spin structure on $W$ and the consequent difficulty of defining half-integer spin fields consistently. Such an obstacle, topological in nature, was pointed out by Gibbons and Hawking and has been used as a selection rule for topology-changing spacetimes \cite{selectionrules}.

Indeed, if $W$ has as boundary the disjoint union of $\Sigma_i$ and $\Sigma_f$, it is possible to introduce a topological invariant $u$ such that $u=0$ when $\partial W$ bounds a connected compact spacetime that admits an $SL(2,\mathbb{C})$ spin structure, and is spacelike with respect to the Lorentzian metric on $W$, and $u =1$ otherwise.

This invariant can be expressed as the mod-2 Kervaire semi-characteristic \cite{CHAMBLIN1994357}
\[
    u(\partial W) = \dim_{\mathbb{Z}_{2}}(H_{0}(\partial W;\mathbb{Z}_{2}) \oplus H_{1}(\partial W; \mathbb{Z}_{2})) \bmod 2,
\]
where $H_{n}(\partial W;\mathbb{Z}_{2})$ are the homology groups; $H_{0}(\partial W;\mathbb{Z}_{2})$ counts the connected components of $\partial W$ and $H_{1}(\partial W; \mathbb{Z}_{2})$ counts its ``one-dimensional holes''. The invariant is additive modulo 2 under disjoint unions,
\[
    u(\Sigma_{1} \sqcup \Sigma_{2}) = u(\Sigma_{1}) + u(\Sigma_{2}) \quad \bmod 2,
\]
and under connected sums it satisfies
\[
    u(\Sigma_{1} \# \Sigma_{2}) = u(\Sigma_{1}) + u(\Sigma_{2}) + 1 \quad \bmod 2.
\]

The $u$-invariant is linked to the usual obstruction to the existence of a spin structure---namely $w_2 = 0$, the vanishing of the second Stiefel–Whitney class---via a theorem of Milnor and Kervaire \cite{Kervaire1963}. This theorem states that the intersection pairing
\[
    h: H_2\left(W ; \mathbb{Z}_2\right) \times H_2\left(W ; \mathbb{Z}_2\right) \longrightarrow \mathbb{Z}_2
\]
satisfies
\[
    \operatorname{rank}(h) = \bigl(u(\partial W) + \chi(W)\bigr) \bmod 2.
\]

When $\partial W$ is everywhere spacelike and $W$ admits a Lorentzian metric, so that $\chi(W) = 0$, we have
\[
    \operatorname{rank}(h) = u(\partial W) \bmod 2,
\]
and, according to \cite{LAWSON1989},
\[
    \operatorname{rank}(h) \bmod 2 = 0 \quad \Longleftrightarrow \quad w_2(W) = 0.
\]

If instead $\partial W$ has changing causal character and so a non-zero kink number, the existence of a Lorentzian metric requires \eqref{eq:generalizedtheorem}; hence $W$ admits a spin structure iff \cite{kinksandtopologychange, CHAMBLIN1994357}
\[
    u(\partial W)=\operatorname{kink}\left(\partial W, g_L\right) \bmod 2.
\]

For the Lorentzian cobordism $W$ describing the wormhole nucleation, the smoothed boundary is $\partial W\cong S^{1}\times S^{2}$, so that $u(\partial W)=0$, whereas $\kink(\partial W,g^{L})=1$; therefore $W$ does not admit a spin structure---as expected, since the $\mathbb{CP}^{2}$ used in the connected sum does not admit one.

Nevertheless, charged spinors can be defined consistently by adopting the generalized spin structure \cite{selectionrules}
\[
    \text{Spin}^{c}(1,3) \cong SL(2,\mathbb{C}) \times_{\mathbb{Z}_2} U(1).
\]

A $\text{Spin}^{c}$ structure requires only that $w_{2}=c$, where $c$ is the mod-2 reduction of an integral class in $H^{2}(W;\mathbb Z)$; all spinorial fields are then $U(1)$-charged and the connection holonomy resolves the ambiguity in defining fermions \cite{HAWKING1978349}. This generalized spin structure always exists on compact $4$-manifolds \cite{Killingback, Whiston1975}.

\section{Wormhole nucleation}
\label{sec:wormholenucleation}

\subsection{Introduction}

In Section~\ref{sec:topologychange}, we established that when a wormhole nucleates within a compact region of spacetime, the enforcement of an everywhere non-degenerate Lorentzian metric requires forming the connected sum between the Morse spacetime $\mathbb{M}$---that is, the cobordism carrying a singular Lorentzian metric derived from the Morse gradient $v^\mu$---and the complex projective space $\mathbb{CP}^{2}$ endowed with its own singular Lorentzian metric. This section elaborates on the geometric and topological aspects of this construction.

Taking the connected sum $W = \mathbb{M} \# \mathbb{CP}^{2}$ requires the removal of a 4-ball $B_1^4$ around the critical point $p \in \mathbb{M}$ of the Morse function, leaving a boundary $3$-sphere $V_1$. Similarly, we remove a $4$-ball $B_2^4$ from $\mathbb{CP}^{2}$ around the unique isolated singular point $q$ of a chosen vector field $w^\mu$, yielding another boundary $3$-sphere $V_2$. The identification of these boundaries is achieved via an orientation-reversing gluing diffeomorphism $h: V_1 \to V_2$, ensuring that the resulting space $W$ inherits a topological-manifold structure through the quotient topology \cite{Milnor1965}.

Following the non-singular construction in Ref.~\cite{Yodzis1972}, the neighborhoods of $V_1$ and $V_2$ in $\mathbb{M}$ and $\mathbb{CP}^{2}$, respectively, are modeled as product neighborhoods $U_1 \cong V_1 \times (-1,0]$ and $U_2 \cong V_2 \times [0,1)$, with local parametrization given by diffeomorphisms $g_1(x_1,t_1)$ for $x_1 \in V_1, \; t_1 \in (-1,0]$, and $g_2(x_2,t_2)$ for $x_2 \in V_2, \; t_2 \in [0,1)$, subject to the condition $g_2(x_2, t_2) = g_2(h(x_1), t_2)$.

The vector field $v^\mu$ restricted to $V_1 \times (-1,0]$ is then transported to $V_2 \times \{0\}$ by successive application of the pushforwards of $g_1$ and $g_2^{-1}$, yielding the transformed field
\[
    v^\mu_{V_2} = \d g_{V_2}^{-1} \circ \d g_{V_1} \circ v_{V_1} \circ g_{V_1} \circ g_{V_2}^{-1},
\]
where the subscripts indicate the restrictions to $V_1$ and $V_2$.

On the other hand, within $\mathbb{CP}^{2}$, the neighborhood $U_2$ can be decomposed into two sub-regions: $U''$, corresponding to the portion mapped by $g^{-1}_2$ for $t_2 \in [1/2,1)$, and $U'_2$, covering $t_2 \in [0,1/2]$. If $w'^\mu$ denotes the restriction of $w^\mu$ to $U''$, the smooth gluing condition requires that $w'^\mu$ can be continuously deformed within $V_2 \times [0,1/2]$ to match $v_{V_2}^\mu$. This deformation is feasible provided that the vector fields $v_{V_2}^\mu$ and $w'^\mu$, normalized with respect to the respective Riemannian metrics and interpreted as maps $S^3 \rightarrow S^3$, possess the same degree \cite{Topologie}.

This is shown to be the case in Section~\ref{sec:gluing}. In this section, we focus on the construction of the two singular Lorentzian metrics, characterizing their geometries and examining the violation of the standard energy conditions.

\subsection{The Morse spacetime}
\label{subsec:morsespacetime}

We begin by constructing the singular Lorentzian metric on the cobordism $\mathbb{M}$. From Theorem~\ref{th:existmorse}, there always exists a Morse function on $\mathbb{M}$; for a nucleating wormhole, this corresponds to a $0$-surgery. In a neighborhood of the critical point $p$ of the Morse function,
there exists a Cartesian coordinate system $x^{\mu} = \{ x^1, x^2, x^3, x^4 \}$ such that the Morse function can be written in the form~\eqref{eq:localmorse}
\[
    f(x^{\mu}) = \frac{1}{2} \left(-(x^1)^2 + (x^2)^2 + (x^3)^2 + (x^4)^2 \right),
\]
where, without loss of generality, we assumed $f(p) = 0$ and rescaled the Morse function to produce the factor of $1/2$. The construction of the singular spacetime $\mathbb{M}$ starting from the Morse gradient field is well known and follows Refs.~\cite{Yodzis1972, Yodzis1973, kostantinov, dowker2002topology}, in which the Lorentzian metric is obtained from Eq.~\eqref{eq:lorentzianfromriemannian}.

The choice of the Riemannian metric is arbitrary. Here we consider the level sets of the Morse function as embedded in $\mathbb{R}^5$, so that a natural Riemannian metric is the one induced on the graph of $f$ from the Euclidean metric in $\mathbb{R}^5$:
\begin{equation}
    g^{R}_{\mu\nu} = \delta_{\mu\nu} + \partial_{\mu}f \partial_{\nu}f,
    \label{eq:riemannianstart}
\end{equation}
where $\delta_{\mu\nu}$ denotes the Euclidean metric on $\mathbb{R}^4$. The singular Lorentzian metric is then constructed from the gradient of the Morse function as
\begin{equation}
    \begin{aligned}
        g^{L}_{\mu\nu} &= g^{R}_{\mu\nu} - 2 \frac{\partial_{\mu}f \partial_{\nu}f}{g_{R}^{\alpha\beta}\partial_{\alpha}f \partial_{\beta}f} = \delta_{\mu\nu} - \left(\frac{n+2}{n}\right) \partial_{\mu}f \partial_{\nu}f,
    \end{aligned}
    \label{eq:firstmorselorentzian}
\end{equation}
where $n = \delta^{\alpha\beta} \partial_{\alpha}f \partial_{\beta}f$ should not be confused with $g_{R}^{\alpha\beta}\partial_{\alpha}f \partial_{\beta}f$.

It is worth noting that the inverse metric is
\[
  g_{L}^{\mu\nu} = g_{R}^{\mu\alpha}g_{R}^{\nu\beta}g^{L}_{\alpha\beta},
\]
and, by construction,
\[
    g^{L}_{\mu\nu} \partial^{\mu}f\partial^{\nu}f = -(g_{R}^{\alpha\beta}\partial_{\alpha}f \partial_{\beta}f) < 0.
\]
Hence, the vector field $\partial^{\mu}f$ is everywhere timelike, and defines the light-cone directions in the spacetime.

Given the spherical symmetry of the Morse function, it is useful to express the Lorentzian metric in spherical coordinates:
\[
    \begin{aligned}
        x^1 &= z, \\
        x^2 &= r \sin(\theta)\sin(\phi), \\
        x^3 &= r \cos(\theta), \\
        x^4 &= r \sin(\theta) \cos(\phi),
    \end{aligned} \quad \quad
    \begin{aligned}
        &z \in \mathbb{R},\\
        &0 < r < \infty, \\
        &0 \le \theta \le \pi, \\
        &0 \le \phi \le 2\pi.
    \end{aligned}
\]

In these coordinates the metric becomes
\begin{equation}
    \d s^2_{\M} = \frac{-U~(z\,\d z - r\,\d r)^2 + (r\,\d z + z\,\d r)^2}{z^2 + r^2} + r^2 \d \Omega_{2}^{2}, \label{eq:morsemetricZR}
\end{equation}
with $\d \Omega_{2}^{2}$ the line element on $S^2$ and $U \equiv U(z,r) = r^2+z^2 +1$. Note that the metric has a curvature singularity at the critical point $z=r=0$ of the Morse function, as expected from the asymptotic behavior of the curvature tensor near the critical point \cite{kostantinov,MYuKonstantinov_1986}:
 \[
     R^\alpha_{\,\, \beta \gamma \delta}, R_{\alpha \beta}, R \simeq \mathcal{O}(1) + \mathcal{O}\left(\frac{1}{g_R^{\alpha\beta}\partial_\alpha f \partial_\beta f}\right).
 \]
However, with the aim of desingularizing the metric via the connected sum, we consider only the points outside an arbitrarily chosen $4$-disk centered at $z=r=0$.

A further change of coordinates, exploiting the simplified form of the Morse function in the $(z,r)$ coordinates, diagonalizes the metric. Define
\[
    t = \frac{1}{2}(r^2 - z^2), \quad \rho = zr.
\]
In these new coordinates Eq.~\eqref{eq:morsemetricZR} becomes
\begin{equation}
    \d s^2_{\M} = \frac{-A\, \d t^2 + \d \rho^2}{2 \sqrt{t^2+\rho^2}} + R\, \d \Omega^2_2,
    \label{eq:diagonalmorsemetric}
\end{equation}
where
\[
    \begin{aligned}
        A \equiv A(t,\rho) &= 1+2 \sqrt{t^2+\rho^2}, \\
        R \equiv R(t,\rho) &= t +\sqrt{t^2+\rho^2}.
    \end{aligned}
\]

These coordinates make it explicit that the hypersurfaces of constant $t$ are the level sets of the Morse function and that the metric~\eqref{eq:diagonalmorsemetric} describes the opening of the wormhole, with a naked singularity at $t=\rho=0$. Indeed, for $t<0$ there is no wormhole and the throat is closed because $R(t,0)=0$; for $t>0$ the wormhole opens with $R(t,0)=2t>0$. Some plots of $R(t,\rho)$ for different values of $t$ are given in Fig.~\ref{fig:figure5}, where for $t<0$ the minimum of $R$ is zero.

\begin{figure}
    \centering
    \includegraphics[width=0.9\linewidth]{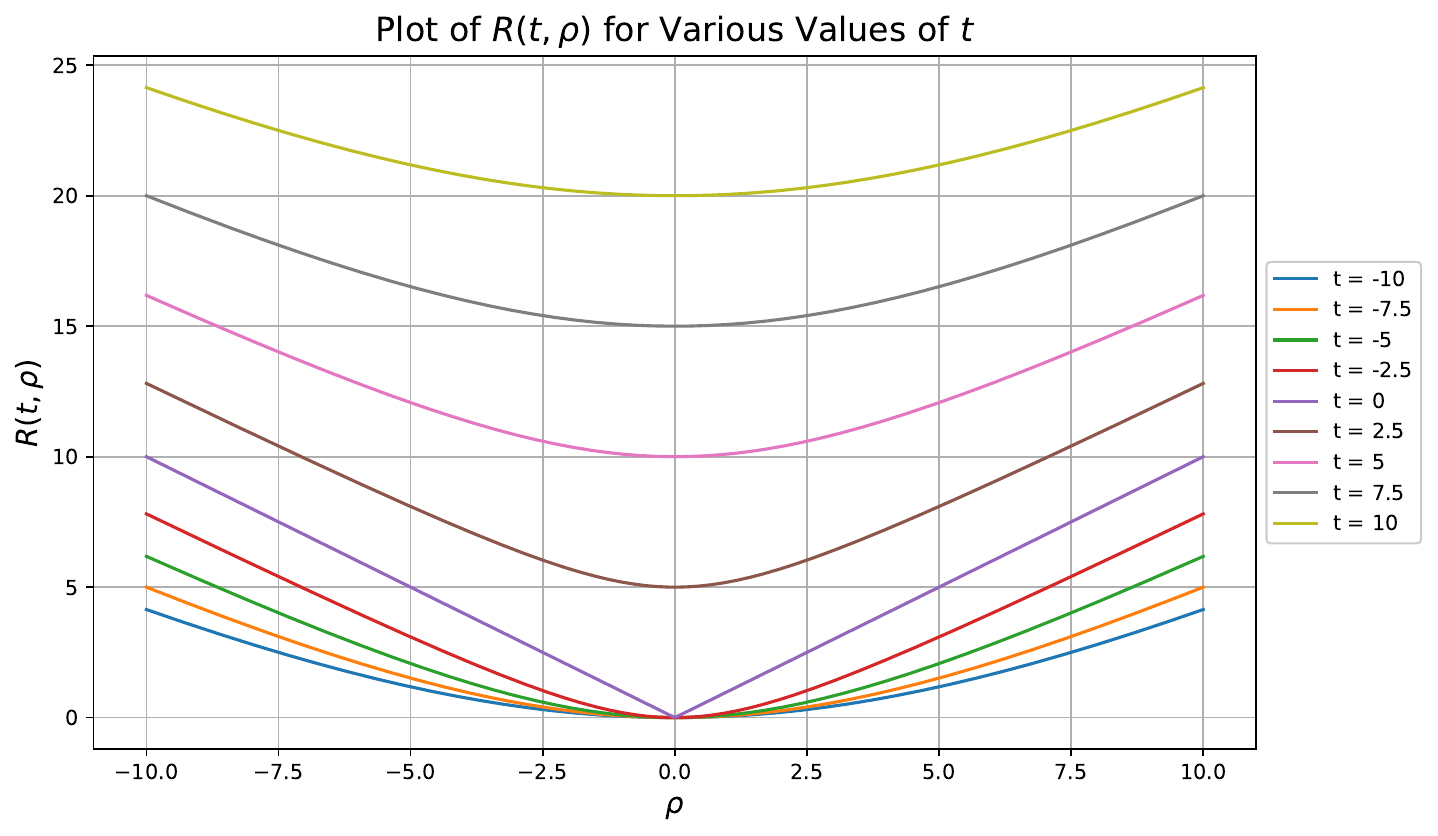}
    \caption{For $t<0$ the wormhole is closed and the minimum of $R(t,\rho)$ occurs at $\rho = 0$ with $R(t,\rho) = 0$; for $t>0$ the wormhole opens and $R(t,0) > 0$.}
    \label{fig:figure5}
\end{figure}

In particular, one can construct an embedding diagram for the induced metric on the hypersurfaces $t=t_c =\text{const}$ and $\theta = \pi /2$:
\[
    \d s^2_{\text{E}} = \frac{\d \rho^2}{2\sqrt{t_c^2+\rho^2}} + R(t_c,\rho)\, \d \phi^2.
\]
Changing variables to $r^2=R(t_c,\rho)$ yields
\[
    \d s^2_{E} = \frac{2(r^2-t_c)}{r^2-2t_c} \d r^2 + r^2\, \d \phi^2.
\]
Comparing with the three-dimensional Euclidean metric in cylindrical coordinates,
\[
    \d s^2 = \d z^2 +\d r^2 +r^2~\d \phi^2,
\]
the embedding surface satisfies
\[
    \left[1+ \left(\frac{\d z}{\d r} \right)^2  \right] = \frac{2(r^2-t_c)}{r^2-2t_c},
\]
whose solution, with integration constant $c_{1}$, is
\[
    z(r) = \pm \sqrt{r^2-2t_c}+c_1.
\]

For different values of $t_c$ we obtain embedding diagrams shown in Fig.~\ref{fig:figure6}, where we set $c_1 = 0$. For $t_c <0$ the wormhole remains closed until the critical level $t_c = 0$ and, thereafter, opens for $t_c >0$.

\begin{figure}
    \begin{subfigure}[b]{0.3\linewidth}
    \includegraphics[width=\linewidth]{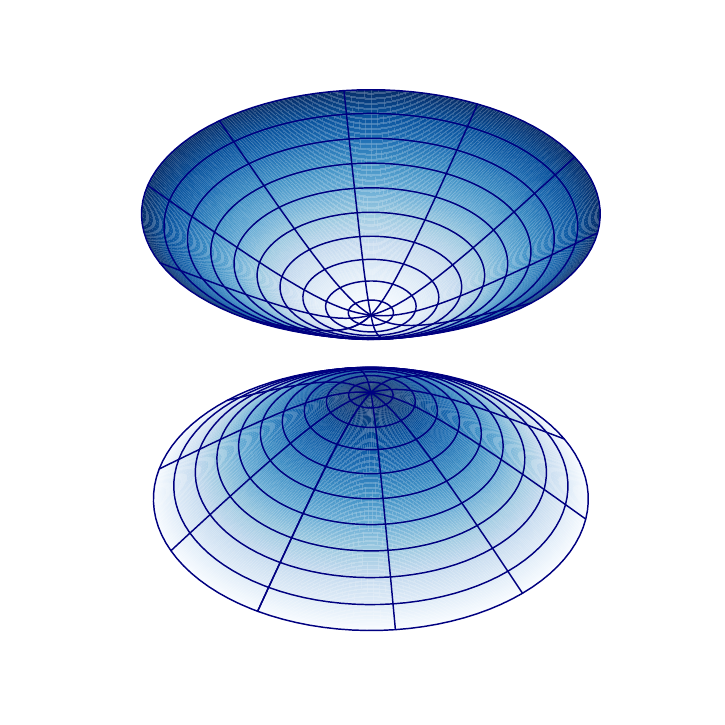}
    \caption{$t_c = -4$}
 \end{subfigure}
 \hfill
 \begin{subfigure}{0.3\linewidth}
 \includegraphics[width=\linewidth]{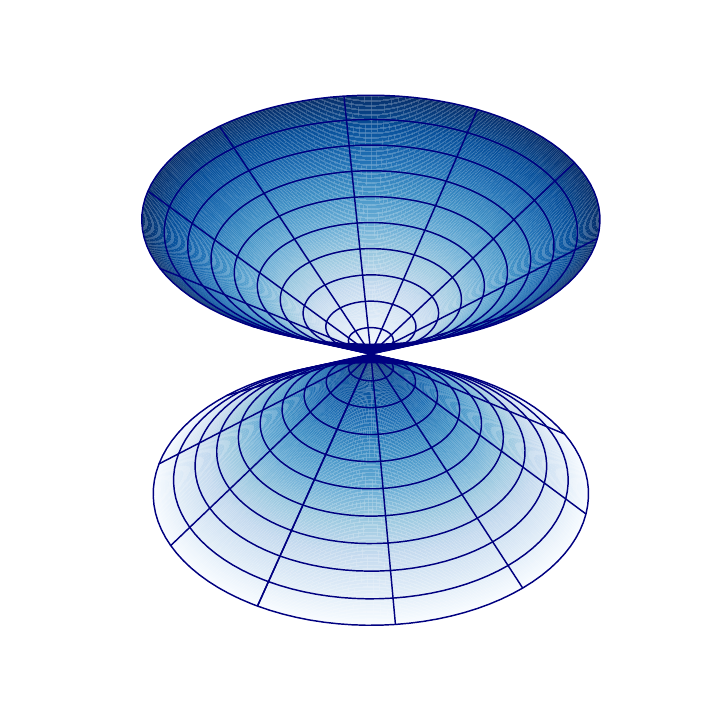}
    \caption{$t_c = 0$}
 \end{subfigure}
 \hfill
 \begin{subfigure}[b]{0.3\linewidth}
 \includegraphics[width=\linewidth]{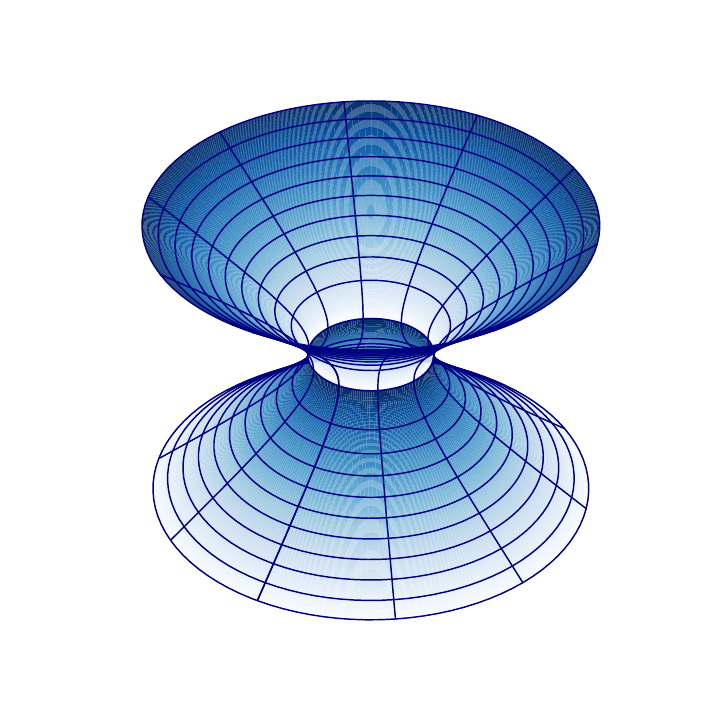}
    \caption{$t_c = 4$}
 \end{subfigure}
\caption{Embedding diagrams for the wormhole metric~\eqref{eq:diagonalmorsemetric}, obtained by rotating the profile $z(r)$ about the $z$–axis. For $t_c<0$ there is no wormhole and the two spatial regions are disconnected. At the critical level $t_c=0$ the throat closes, whereas for $t_c>0$ the wormhole opens with a dynamic throat.}
\label{fig:figure6}
\end{figure}

To analyze the behavior of geodesics in this spacetime we note that it possesses three Killing vectors, arising from the spherical symmetry; they are preserved by the Lie bracket between $\partial^\mu f$ and the Killing vectors of $g^{R}_{\mu\nu}$ (see Theorem~\ref{th:killingvectors}):
\[
    \begin{aligned}
        \xi_{1}^{\mu} &= \sin(\phi) \partial_\theta + \cos(\phi)\cot(\theta)\partial_\phi, \\
        \xi_{2}^{\mu} &= \cos(\phi) \partial_\theta - \sin(\phi)\cot(\theta)\partial_\phi, \\
        \xi_{3}^{\mu} &= \partial_\phi.
    \end{aligned}
\]

The geodesics satisfy the Euler-Lagrange equations derived from the Lagrangian
\[
\mathcal{L} = \frac{1}{2 \sqrt{t^2 + \rho^2}}\left(-A\dot{t}^2 + \dot{\rho}^2 \right)+R \left(\dot{\theta}^2+ \dot{\phi}^2 \sin^2\theta \right) = \epsilon,
\]
where $\epsilon$ is zero for lightlike geodesics and negative for timelike geodesics, and a dot denotes differentiation with respect to an affine parameter~$\lambda$.

Owing to spherical symmetry, the conserved quantities associated with $\xi^{\mu}_{1}$ and $\xi^{\mu}_{2}$ fix the orbital plane, which we may take as $\theta = \pi/2$, while $\xi^\mu_3$ yields the conserved angular momentum
\[
L = \xi_\mu^3 \frac{dx^\mu}{d \lambda} = R\dot{\phi} = \mathrm{const}.
\]
Substituting these conserved quantities, the reduced Lagrangian becomes
\[
\mathcal{L} = \frac{1}{2 \sqrt{t^2 + \rho^2}}\left(-A\dot{t}^2 + \dot{\rho}^2 \right)+\frac{L^2}{R} = \epsilon.
\]

Taking variational derivatives of $\mathcal{L}$ and employing the constraint $\mathcal{L}=\epsilon$, one obtains the geodesic equations
\begin{align*}
    \frac{t\left(\epsilon -L^2/R +\dot{t}^2 \right)}{t^2+\rho^2}+\frac{L^2 \partial_t R}{R^2} &= -\frac{\d}{\d \lambda}\left(\frac{\dot{t} A}{\sqrt{t^2+\rho^2}}\right), \\
    \frac{\rho\left(\epsilon -L^2/R +\dot{t}^2\right)}{t^2+\rho^2}+\frac{L^2 \partial_\rho R}{R^2} &= -\frac{\d}{\d \lambda}\left(\frac{\dot{\rho}}{\sqrt{t^2+\rho^2}}\right).
\end{align*}

We now examine radial geodesics, for which $L=0$, that traverse the wormhole after its formation ($t>0$). The geodesic equations then simplify to
\begin{align*}
    \frac{t(\epsilon +\dot{t}^2)}{t^2+\rho^2} &= -\frac{\d}{\d \lambda}\left(\frac{\dot{t} A}{\sqrt{t^2+\rho^2}}\right), \\
    \frac{\rho(\epsilon +\dot{t}^2)}{t^2+\rho^2} &= -\frac{\d}{\d \lambda}\left(\frac{\dot{\rho}}{\sqrt{t^2+\rho^2}}\right),
\end{align*}
and, for $t^2 + \rho^2 \ne 0$, reduce further to
\begin{subequations}
\begin{align}
    \ddot{t}(t^2 + \rho^2)A-\dot{t}(t\dot{t}+\rho \dot{\rho}) + t(\epsilon +\dot{t}^2) \sqrt{t^2+\rho^2}=0, \label{eq:geodesic1}\\
    \ddot{\rho}(t^2 + \rho^2) -\dot{\rho}(t\dot{t}+\rho \dot{\rho}) + \rho(\epsilon + \dot{t}^2) \sqrt{t^2+\rho^2}=0. \label{eq:geodesic2}
\end{align}
\end{subequations}

For timelike geodesics, we set $\epsilon = -1$ and solve numerically. Choosing the initial conditions $\rho(0) = 5$ and $t(0) = 1$, the initial velocity must satisfy
\[
-(1+2\sqrt{26})~\dot{t}^2_0 + \dot{\rho}^2_0 = -2\sqrt{26}.
\]
Taking the positive time direction, $\dot{t}_{0}=1$, gives $\dot{\rho}_0 = \pm 1$, representing geodesics either entering the wormhole or escaping to infinity. Because we are interested in the incoming case, we choose $\dot{\rho}_0 = -1$; the numerical result is plotted in Fig.~\ref{fig:figure7}.

\begin{figure}
    \centering
    \includegraphics[width=1\linewidth]{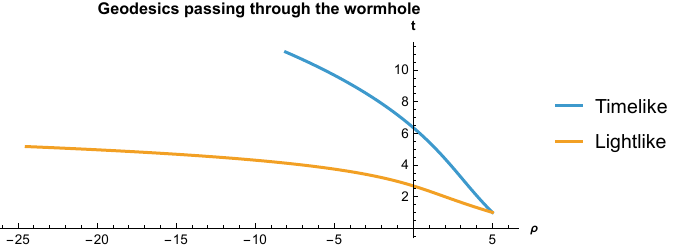}
    \caption{Numerical solutions of Eqs.~\eqref{eq:geodesic1} and \eqref{eq:geodesic2} for timelike and lightlike geodesics passing through the wormhole. The integration is performed over the affine interval $\lambda\in[0,10]$ with $t(0)=1$ and $\rho(0)=5$.\newline
    \textit{Blue curve:} Timelike geodesic with $\dot{t}_{0}=1$, $\dot{\rho}_{0}=-1$.\newline
    \textit{Orange curve:} Lightlike geodesic with $\dot{t}_{0}=1/2$, $\dot{\rho}_{0}=-\sqrt{1+2\sqrt{26}}/2$.}
    \label{fig:figure7}
\end{figure}

A similar procedure applies to lightlike geodesics, with $\epsilon = 0$. With the same initial positions and $\dot{t}_{0}=1/2$, the null condition
\[
    -(1+2\sqrt{26})~\dot{t}^2_0 + \dot{\rho}^2_0 = 0
\]
selects $\dot{\rho}_0 = - (\sqrt{1+2\sqrt{26}})/2$, and the resulting trajectory is also plotted in Fig.~\ref{fig:figure7}.

Before briefly discussing violations of the energy conditions for this dynamical wormhole, we want to identify, in a more rigorous way, the fundamental properties of the wormhole. These properties should be intrinsic and insensitive to the global structure of the spacetime. A general approach to this problem is given by the flare-out conditions presented in Ref.~\cite{FlareOut}. Indeed, a tetrad can be read off from the metric~\eqref{eq:diagonalmorsemetric}:
\[
\begin{aligned}
    \textbf{e}_t = \frac{\sqrt{2A}}{2(t^2 + \rho^2)^{\frac{1}{4}}}\d t, \quad \textbf{e}_\rho = \frac{\sqrt{2}}{2(t^2 + \rho^2)^{\frac{1}{4}}}\d \rho \\
    \textbf{e}_\theta = \sqrt{R}\,\d \theta, \quad \textbf{e}_\phi = \sqrt{R}\sin\theta\,\d \phi,
\end{aligned}
\]
from which a null tetrad, used in the Newman-Penrose formalism, can be constructed as
\[
\begin{aligned}
    \textbf{l}_+ &= \frac{1}{\sqrt{2}}(\textbf{e}_t + \textbf{e}_\rho), \quad \textbf{l}_{-} = \frac{1}{\sqrt{2}}(\textbf{e}_t - \textbf{e}_\rho),\\
    \textbf{m}&=\frac{1}{\sqrt{2}}(\textbf{e}_\theta + \i\textbf{e}_\phi), \quad \bar{\textbf{m}}=\frac{1}{\sqrt{2}}(\textbf{e}_\theta - \i\textbf{e}_\phi),
\end{aligned}
\]
where $l^\mu_+$ and $l^\mu_{-}$ are real null vectors and $m^\mu$ and $\bar{m}^\mu$ are complex-conjugate null vectors. Using this null tetrad, one can deduce that the Morse spacetime is of Petrov type D.

This implies that there exist two principal null directions:
\[
    k^{\mu}_{1,2} = (t^2+\rho^2)^{\frac{1}{4}} \left(-\frac{\partial_{t}}{\sqrt{A}} \pm \partial_\rho \right),
\]
which correspond to $k_1^{\mu} =l_{+}^{\mu}$ and $k_{2}^{\mu} = l_{-}^{\mu}$ and satisfy
\[
   l^{\pm}_{[\mu} C_{\nu]\rho\sigma[\alpha}l^{\pm}_{\beta]}\, l_{\pm}^\rho l_{\pm}^\sigma  = 0,
\]
where $C_{\mu\nu\rho \sigma}$ is the Weyl tensor. The principal null directions describe null-geodesic congruences:
\[
   l_{\pm}^\nu\nabla_\nu l_{\pm}^\mu = f_{\pm}(t,\rho) l^\mu_{\pm},
\]
where
\[
    f_{\pm}(t,\rho) =\frac{t\sqrt{A} \pm (t^2+\rho^2) \partial_\rho A -\rho A}{2A(t^2+\rho^2)^\frac{3}{4}}.
\]

Defining the spatial projector
\[
    \gamma_{\alpha \beta}\,\d x^\alpha \d x^\beta = R(t,\rho)(\d \theta^2 + \sin(\theta)^2 \d \phi^2),
\]
the following relations hold:
\begin{align*}
    &l_{+}^\mu l^{+}_\mu = l_{-}^\mu l^{-}_\mu = 0,\quad
    l_{+}^\mu l^{-}_\mu = l_{-}^\mu l^{+}_\mu = -1,\\
    &l_{+}^\alpha \gamma_{\alpha \beta} = l_{-}^\alpha \gamma_{\alpha \beta} = 0,\quad
    \gamma^{\alpha}_{\delta}\gamma^{\delta \sigma} = \gamma^{\alpha \sigma}.
\end{align*}
Thus the Lorentzian metric~\eqref{eq:diagonalmorsemetric} can be decomposed as
\[
    \d s^{2}_{\M} = (\gamma_{\mu\nu} -l^{+}_{\mu}l^{-}_{\nu} - l^{-}_{\mu}l^{+}_{\nu})~\d x^{\mu} \d x^\nu.
\]

The wormhole throats are defined, for each congruence, as closed minimal $2$-surfaces $\Sigma_{\pm}$ (with $+$ corresponding to $l_{+}^\mu$ and $-$ to $l_{-}^\mu$). The condition for $\Sigma_{\pm}$ to be minimal is obtained by extremizing the area
\[
    A\left(\Sigma_{\pm}\right)=\int_{\Sigma_{\pm}} \sqrt{\gamma}~\d ^2x,
\]
so that $\delta (A) = 0$. This variation equals the expansion scalar for each null congruence \cite{FlareOut}:
\[
    \theta_{\pm} = \gamma^{\alpha \beta}\nabla_{\alpha}l_{\beta}^{\pm} = \frac{1}{2} \gamma^{a b} \frac{d \gamma_{a b}}{d u_{\pm}} = 0,
\]
where $u_\pm$ are affine parameters. Hence the throats are the set of points where the expansion scalar of both null congruences vanishes. In the dynamical situation two throats are generally expected that coincide in the static limit. We show that this is indeed the case for the metric~\eqref{eq:diagonalmorsemetric}; specifically, for $t<0$ the two minimal surfaces coalesce (zero area, $R=0$), whereas for $t>0$ they separate.

The vanishing of $\theta_{\pm}$ is not, however, sufficient to guarantee that $\Sigma_{\pm}$ are minimal; one must also have $\delta^{2}A\ge0$, which gives the flare-out condition
\[
    \frac{\d \theta_{\pm}}{\d u_{\pm}} \ge 0.
\]
For a generic wormhole the stronger averaged flare-out condition should hold:
\begin{equation}
    \int_{\Sigma_u} \sqrt{\gamma} \operatorname{sgn}\left(\frac{\d \theta}{\d u}\right) \d^2 x>0.
    \label{eq:strongflareout}
\end{equation}

We apply the conditions just stated to our case, using the null principal directions previously computed. The expansion scalar is then
\begin{equation}
    \theta_{\pm} = \frac{\left(\rho^{2}+t^{2}\right)^{1/4} \left(\pm \sqrt{A} \partial_{\rho} R - \partial_{t} R\right)}{R \sqrt{A}}.
    \label{eq:nullexpansion}
\end{equation}

Assuming $t^{2}+\rho^{2} \neq 0$, the throats are located by solving $\theta_{\pm}=0$; the solution is displayed in Fig.~\ref{fig:figure8}. As anticipated, for $t>0$ the two throats bifurcate with $\rho_{\pm} \neq 0$, while for $t<0$ they coalesce at $\rho_{\pm} = 0$.

\begin{figure}
    \centering
    \includegraphics[width=1\linewidth]{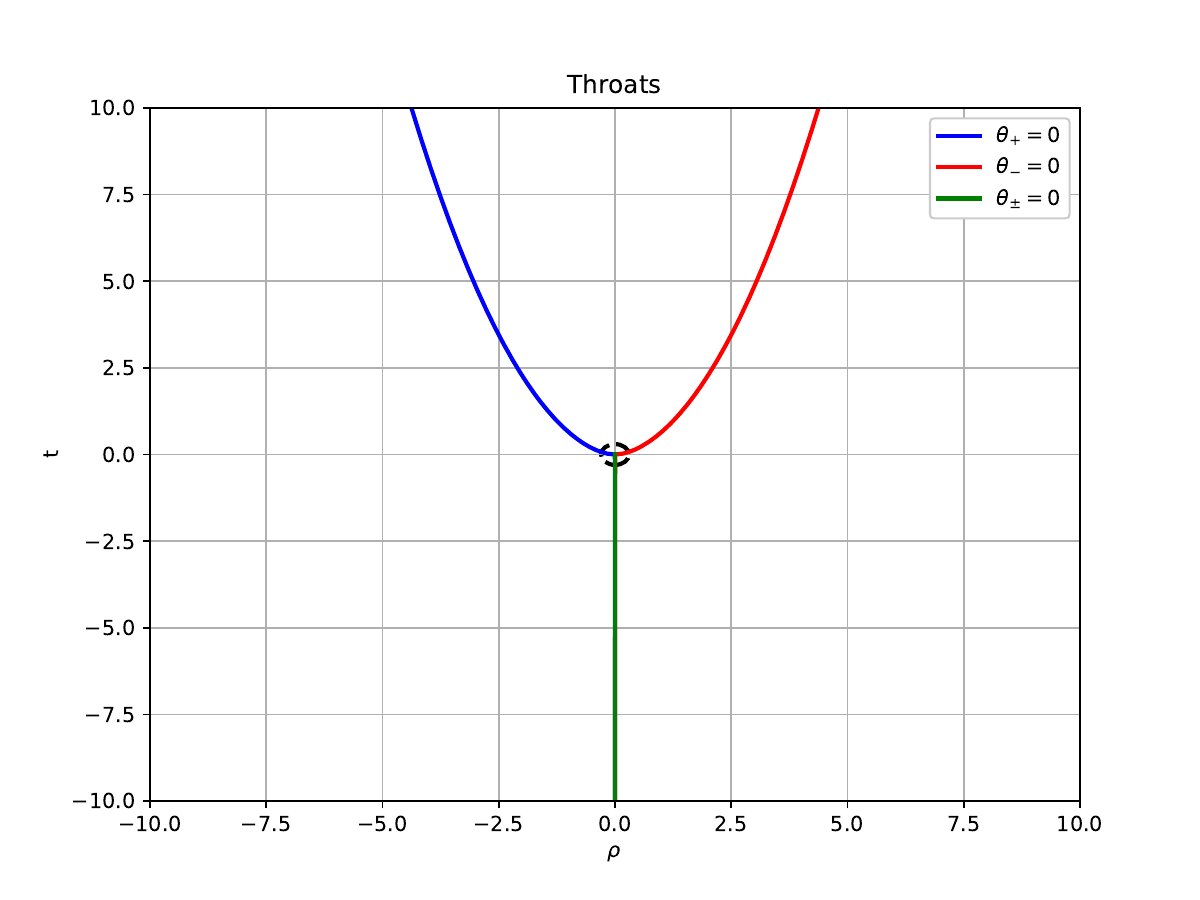}
    \caption{The region in which the expansion scalars $\theta_{\pm}$ vanish. The solid curve, produced numerically by the condition~\eqref{eq:nullexpansion}, shows that for $t>0$ the two wormhole throats bifurcate, whereas for $t<0$ they coalesce (with $R=0$ at $\rho=0$). Note that the bifurcation point, which coincides with the naked singularity, is not included in the spacetime.}
    \label{fig:figure8}
\end{figure}

At points where \eqref{eq:nullexpansion} holds, the null Raychaudhuri equation reduces to
\[
    \left.\frac{\d \theta_{\pm}}{\d u_{\pm}}\right|_{\theta_{\pm} = 0} = -R_{\mu\nu}l_{\pm}^{\mu}l_{\pm}^{\nu}.
\]
Here the term proportional to $\theta_{\pm}^2$ vanishes; because $l^\mu_{\pm}$ are surface-orthogonal, the vorticity tensor $\omega_{\mu\nu}$ is zero, and the shear tensor also vanishes at $\theta_{\pm}=0$. We consequently obtain that
\[
   \left.\frac{\d \theta_{\pm}}{\d u_{\pm}}\right|_{\theta_{\pm} = 0} = \frac{ At^2 \pm t\rho \sqrt{A}+ \rho^2 \sqrt{t^2+\rho^2}}{A^2R~(t^2+\rho^2)},
\]
which is strictly positive in an open neighborhood of the two throats at $t>0$, thereby satisfying the averaged flare-out condition \eqref{eq:strongflareout}. This expression is undefined for $t<0$, confirming that no wormhole is present at negative times.

We conclude this section with a brief analysis of the violations of the energy conditions. The analysis is conducted in the coordinates $\{ z,r,\theta, \phi \}$, in which the metric takes the form~\eqref{eq:morsemetricZR}. This yields a slight simplification in the expression that classifies the Hawking–Ellis types of the stress–energy tensor. By defining the function $F\equiv F(z,r) = z^2 +r^2$ a tetrad for the metric~\eqref{eq:morsemetricZR} is given by
\begin{subequations}
    \begin{align}
        \mathbf{e}_{0} &= \sqrt{\frac{F+1}{F}} (z \d z -r \d r),  \label{eq:tetrad_morse_1}\\[1ex]
        \mathbf{e}_{1} &= \frac{r \d z + z \d r}{\sqrt{F}}, \label{eq:tetrad_morse_2}\\[1ex]
        \mathbf{e}_{2} &= r\,\d \theta, \qquad
        \mathbf{e}_{3} = r \sin(\theta)\,\d \phi. \label{eq:tetrad_morse_4}
    \end{align}
\end{subequations}

To classify the stress-energy tensor into the different Hawking-Ellis types, one needs to investigate its eigenvalues, as shown in Refs.~\cite{Hawking_Ellis_1973,Martín–Moruno2017,Maeda_2022}. In many cases, however, as pointed out in Ref.~\cite{Maeda_2022}, the starting tetrad is not the canonical one in which the stress-energy tensor type is explicit, and a Lorentz transformation may be needed to bring it to canonical form.

When spherical symmetry is present, the Hawking-Ellis type can, however, be determined from any tetrad for which the non-zero components of $T^{ab}$ are $T^{00}$, $T^{10}$, $T^{11}$, and $T^{22} = T^{33}$ \cite{Maeda_2022}:
\begin{equation}
    \begin{array}{lll}
        \left(T^{00}+T^{11}\right)^2>4\left(T^{01}\right)^2 & \Rightarrow & \text { Type I, } \\
        \left(T^{00}+T^{11}\right)^2=4\left(T^{01}\right)^2 & \Rightarrow & \text { Type II, } \\
        \left(T^{00}+T^{11}\right)^2<4\left(T^{01}\right)^2 & \Rightarrow & \text { Type IV. }
    \end{array}
\end{equation}
In particular, a spherically symmetric spacetime cannot be of type III.

Using the tetrad~\eqref{eq:tetrad_morse_1}–\eqref{eq:tetrad_morse_4}, the test reduces to checking the sign of
\begin{align*}
    &r^{8} + 4 r^{6} - 2\left(z^{4}+6z^{2}-4\right) r^{4} \\
    &\quad - 4\left(5z^{4}+2z^{2}-2\right) r^{2} + \left(z^{4}-2z^{2}-2\right)^{2}.
\end{align*}

This expression is positive, zero, or negative, corresponding to types I, II, and IV, respectively. These regions are shown in Fig.~\ref{fig:figure9}. Thus, the spacetime~\eqref{eq:morsemetricZR}, having regions where the stress–energy tensor is type IV, violates all the standard energy conditions. We stress that this does not mean that the violations originate solely from the type IV region; energy-condition violations also occur in the type I and type II regions.

\begin{figure}
    \includegraphics[width=\linewidth]{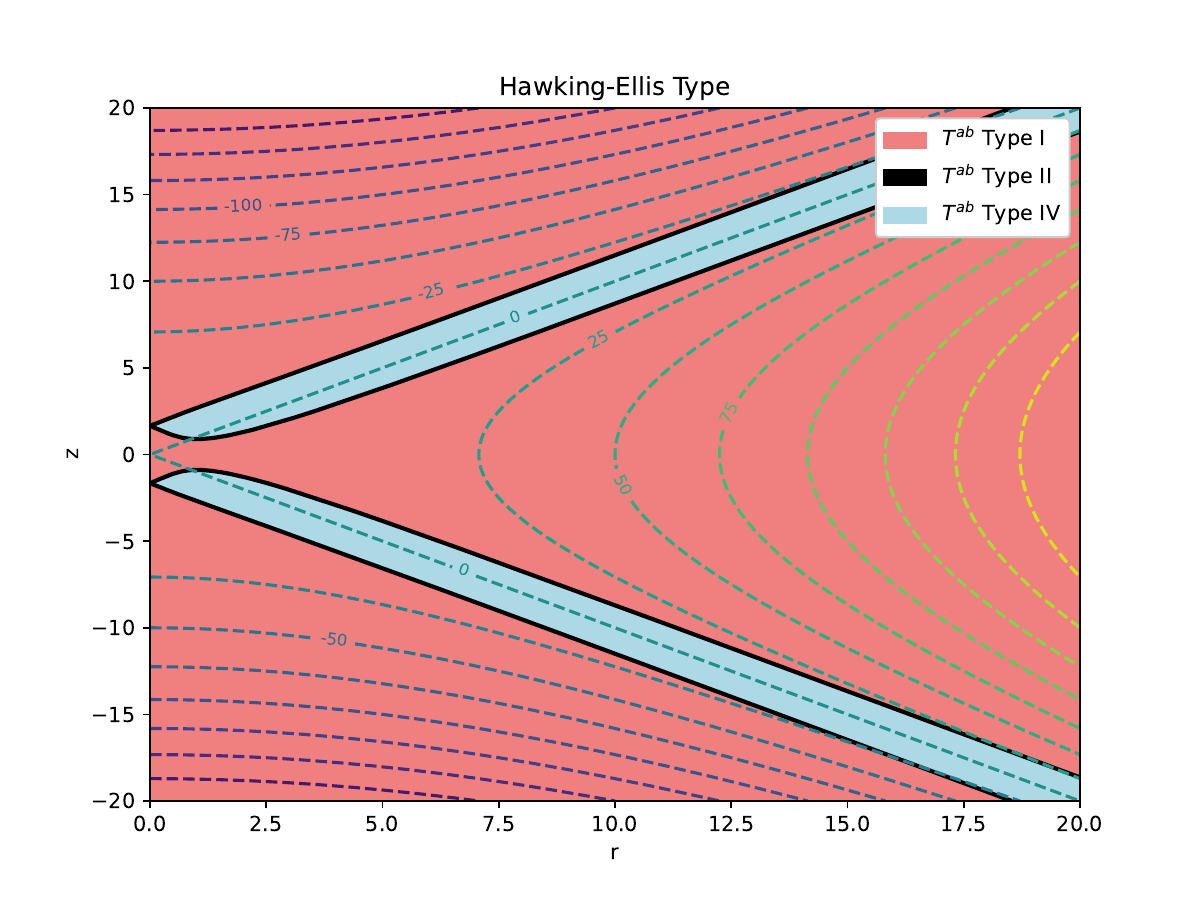}
    \caption{The stress–energy tensor changes type in different regions of spacetime. The pastel-red regions correspond to type I, while the light-blue regions correspond to type IV\@. At the boundary the tensor is type II\@. The graph also depicts the level sets of the Morse function in $(z,r)$ coordinates, revealing that the energy‑condition violations associated with type-IV regions occur on the critical level set.}
    \label{fig:figure9}
\end{figure}

Finally, the regions highlighted in Fig.~\ref{fig:figure9} are also those where the Segre type changes from $[Z\bar{Z}(11)]$ (type IV) to $[(1,1)11]$ (type I); on the boundary the Segre type is $[2,(11)]$.

\subsection{A singular Lorentzian metric on the complex projective plane}
\label{subsec:complexprojectiveplane}

The complex projective plane $\mathbb{CP}^{2}$ is a two-dimensional complex manifold constructed by identifying the triples $(Z_1, Z_2, Z_3) \in \mathbb{C}^3$ under the equivalence relation
\[ (Z_1, Z_2, Z_3 ) \sim  (\lambda Z_1, \lambda Z_2, \lambda Z_3 ), \quad \lambda \in \mathbb{C}\setminus\{ 0 \}. \]

This relation identifies all points (except the origin) lying on the same complex line through the origin of $\mathbb{C}^3$. From a topological point of view, $\mathbb{CP}^{2}$ can be obtained by attaching a 4-disk $D^4$ to a 2-sphere $S^2$ along its boundary via the Hopf projection $S^3 \rightarrow S^2$.

The coordinates $(Z_1, Z_2, Z_3)$, called homogeneous coordinates, are redundant and can be replaced with the following affine coordinates:
\[
    \begin{aligned}
        U_1: \quad (z_1, z_2) &= \left( \frac{Z_1}{Z_3}, \frac{Z_2}{Z_3}\right) \quad \mathrm{for} \quad Z_3 \ne 0, \\
        U_2: \quad (\zeta_1, \zeta_2) &= \left( \frac{Z_2}{Z_1}, \frac{Z_3}{Z_1}\right) \quad \mathrm{for} \quad Z_1 \ne 0, \\
        U_3: \quad (\xi_1, \xi_2) &= \left( \frac{Z_1}{Z_2}, \frac{Z_3}{Z_2}\right) \quad \mathrm{for} \quad Z_2 \ne 0,
    \end{aligned}
\]
which jointly cover $\mathbb{CP}^{2}$. Each open set is diffeomorphic to $\mathbb{C}^2$---hence $\mathbb{R}^4$---so that $\mathbb{CP}^{2}$ can be regarded as a compactification of $\mathbb{R}^4$ by the sphere $\mathbb{CP}^1 \cong S^2$ at infinity. More details on the structure of the complex projective plane and its use in physics as a gravitational instanton can be found in the excellent reference~\cite{Gibbons1978}.

Here, we are interested in constructing a vector field on $\mathbb{CP}^{2}$ that has only one critical point. Such a field always exists, as established by a theorem of Hopf\footnote{Translated from German. Theorem II, p.~550 in \cite{Topologie}.}:

\begin{theorem}
    In any differentiable manifold $M$ there exist vector fields with a finite number of singularities and there always exists a vector field with exactly one singularity.
\end{theorem}

By the Poincaré-Hopf theorem, this vector field must necessarily have index $+3$ at its isolated singularity. To obtain it explicitly, in the absence of any differential equations, we exploit the transition functions between the local charts $U_i$. Roughly speaking, by specifying a non-singular vector field on one chart---say $U_1$---we force the transition functions to introduce the required singularity on the other charts, in accordance with the Poincaré-Hopf theorem.

This procedure is particularly straightforward when there are only two charts, as in the case of $\mathbb{CP}^1$; nevertheless, it can be extended to $\mathbb{CP}^{2}$ by considering how rational vector fields transform between the charts. As a first step, we note that if one considers a polynomial vector field on $U_1$,
\[
    a_{ij}z_1^{i} z_2^{j} \partial_{z_1} + b_{ij}z_1^{i}z_2^{j} \partial_{z_2},
\]
then, in order to avoid poles or indeterminate points in the other charts---so that the resulting singularity is isolated---the degrees of the polynomial coefficients must be at most $2$. Under a change of coordinates to $U_2$, one obtains
\[
    -\left[a_{ij} \zeta_1^{2-i-j} \zeta_2^j\right] \partial_{\zeta_1} -\left[\zeta_1^{1-i-j} \zeta_2^j \left(a_{i j} \zeta_2+b_{i j}\right)\right] \partial_{\zeta_2}.
\]
This observation greatly simplifies the search for the desired vector field. For example, one may choose
\begin{equation}
    \begin{aligned}
        & U_1: \quad \left(z_1^2+z_2\right) \partial_{z_1}+z_1 z_2 \partial_{z_2}, \\
        & U_2: \quad -\zeta_1^2 \partial_{\zeta_1}-(1+\zeta_1 \zeta_2) \partial_{\zeta_2}, \\
        & U_3: \quad \partial_{\xi_1}-\xi_1\partial_{\xi_2},
    \end{aligned}
\end{equation}
so that the only singularity is the isolated one in $U_1$ at $(z_1, z_2) = (0,0)$. We stress that the vector field obtained in this way is not unique; other choices are possible.

We remark that this vector field defines a singular holomorphic foliation of $\mathbb{CP}^{2}$. In particular, the degree-$2$ part of the vector field in $U_1$,
\[
    z_1 \left( z_1 \partial_{z_1} + z_2 \partial_{z_2}\right),
\]
enables us to deduce, by Proposition 7.1 in Ref.~\cite{Lopez2000}, that the line at infinity $\mathbb{CP}^{1}_{\infty}$ is not a leaf of the singular holomorphic foliation.

Finally, since we treat $\mathbb{CP}^{2}$ as a real manifold, it is useful to extract the real part of the vector field:
\begin{widetext}
    \begin{equation}
       w^\mu =
       \begin{cases}
           {\begin{aligned} U_1:  (x^2 - y^2 +z)\, \partial_{x} + (2xy + t) \, \partial_y  +(xz-ty) \, \partial_z + (tx +yz) \, \partial_t, \end{aligned}} \\
           {\begin{aligned}U_2 :  -(a^2-b^2)\, \partial_{a} -2ab \, \partial_b  -(ac-bd+1) \, \partial_c -(ad+bc) \, \partial_d, \end{aligned}} \\
            U_3 : {\partial_\alpha - \alpha \, \partial_\gamma -\beta \, \partial_\delta,}
        \end{cases}
        \label{eq:vectorfieldprojective}
    \end{equation}
\end{widetext}
which has the same property as the complex vector field, namely a single isolated singularity of index $+3$ in the first chart $U_1$. We have chosen Cartesian coordinates on $\mathbb{CP}^{2}$ following Ref.~\cite{PhysRevLett.37.1251}: $\{ x,y,z,t\}$ on $U_1$, $\{ a,b,c,d \}$ on $U_2$, and $\{ \alpha, \beta, \gamma, \delta \}$ on $U_3$. Several two-dimensional projections of the integral curves of the vector field~\eqref{eq:vectorfieldprojective} $U_1$ are shown in Fig.~\ref{fig:figure10}.

\begin{figure*}
    \begin{subfigure}{0.3\linewidth}
        \centering
        \includegraphics[width=\linewidth]{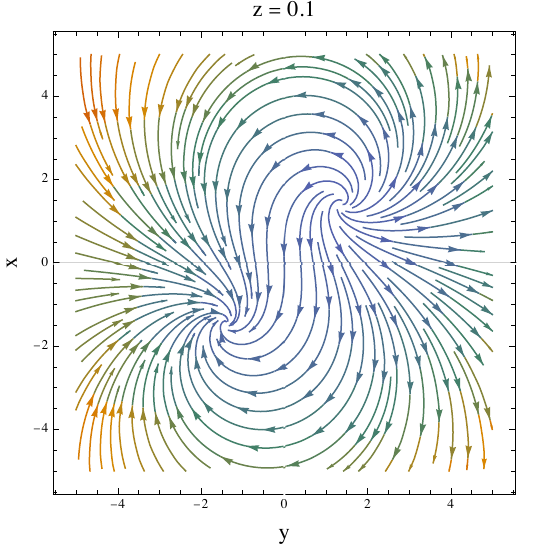}
        \caption{$t = -4$}
    \end{subfigure}
    \hfill
    \begin{subfigure}{0.3\linewidth}
        \centering
        \includegraphics[width=\linewidth]{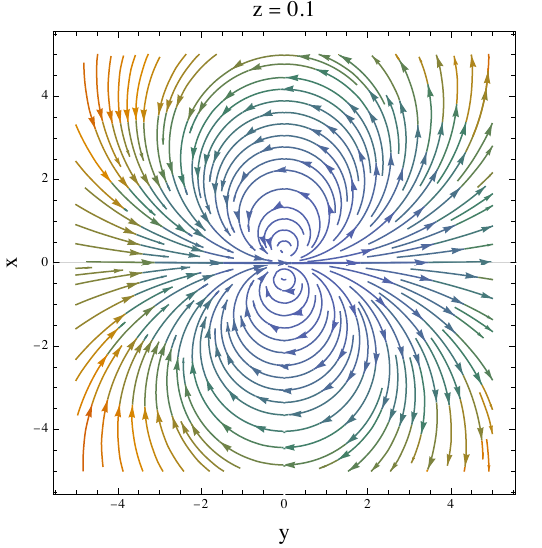}
        \caption{$t = 0$}
    \end{subfigure}
    \hfill
    \begin{subfigure}{0.3\linewidth}
        \centering
        \includegraphics[width=\linewidth]{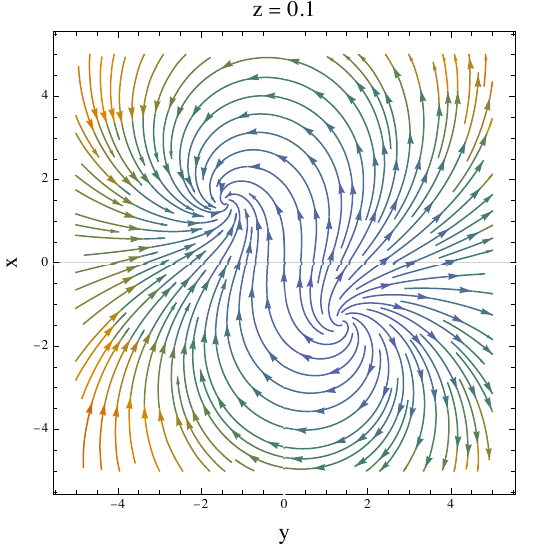}
        \caption{$t = 4$}
    \end{subfigure}
    \begin{subfigure}{0.3\linewidth}
        \centering
        \includegraphics[width=\linewidth]{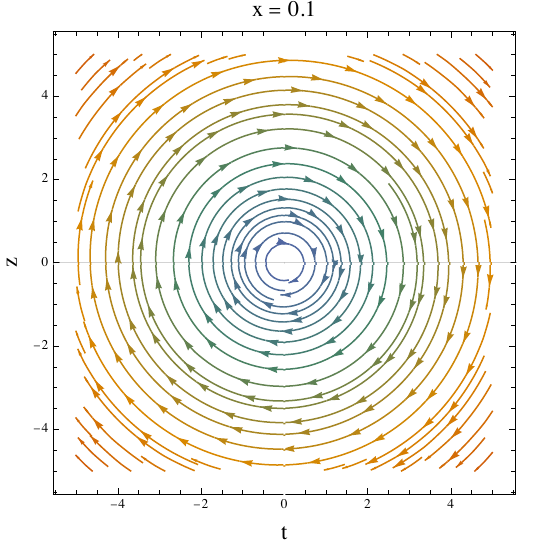}
        \caption{$y = -4$}
    \end{subfigure}
    \hfill
    \begin{subfigure}{0.3\linewidth}
        \centering
        \includegraphics[width=\linewidth]{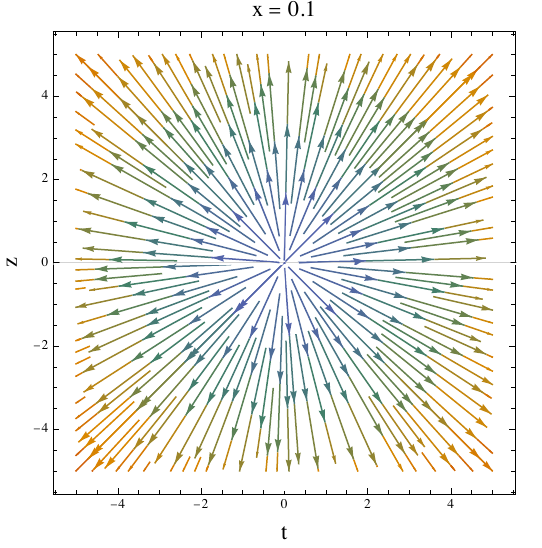}
        \caption{$y = 0$}
    \end{subfigure}
    \hfill
    \begin{subfigure}{0.3\linewidth}
        \centering
        \includegraphics[width=\linewidth]{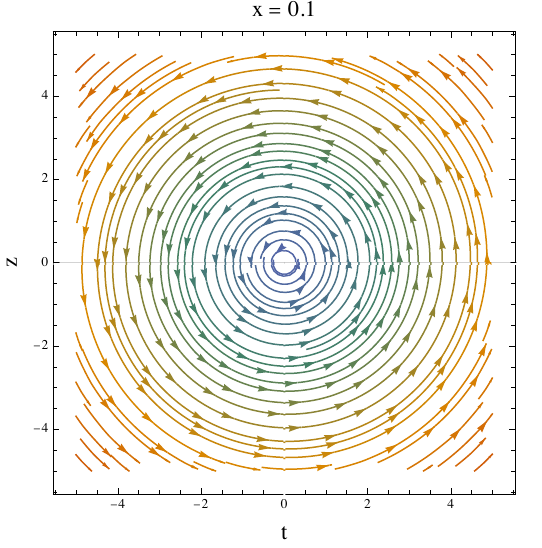}
        \caption{$y = 4$}
    \end{subfigure}
    \caption{Two-dimensional streamline plots of the vector field $w^\mu$ in the $U_1$ chart. The color scale indicates the magnitude of the vector field, with warmer colors representing higher values. (a)–(c): Projections onto the $xy$-plane for different values of the $t$ coordinate. In these plots $z$ is fixed at $0.1$ to obtain a two-dimensional view. The number of singularities changes from two (each of index $+1$ at $t = -4$ and $t = 4$) to a single singularity at the origin (of index $+2$ at $t = 0$). (d)–(f): Projections onto the $zt$-plane for different values of $y$, with $x$ fixed at $0.1$. In all these cases the singularity is located at the origin and has index $+1$.}
    \label{fig:figure10}
\end{figure*}

The singular Lorentzian metric on $\mathbb{CP}^{2}$ is therefore constructed from the Riemannian Fubini-Study metric and the vector field~\eqref{eq:vectorfieldprojective}:
\begin{align}
    g^L_{\mu\nu} = g^{FS}_{\mu\nu} - 2 \frac{w_\mu w_\nu}{g^{FS}_{\alpha\beta} w^\alpha w^\beta},
    \label{eq:LorentzianmetricCP2}\\
    g^{FS}_{\mu\nu} \equiv \frac{6/\Lambda}{1+r^2} \left[\delta_{\mu\nu} - \frac{x^\mu x^\nu + \tilde{x}^\mu \tilde{x}^\nu}{1+r^2} \right],
    \label{eq:FubiniStudy}
\end{align}
where $\Lambda>0$ is a constant and, in the chart $U_1$,
\begin{align*}
    r^2 = x^2 + y^2 + z^2 +t^2, \\
    (\tilde{x}^1, \tilde{x}^2, \tilde{x}^3, \tilde{x}^4) = (y, -x, t, -z).
\end{align*}
Here $\Lambda$ plays the role of a cosmological constant, since the Fubini-Study metric is a Riemannian solution of the Einstein equations with positive $\Lambda$. A coordinate system that simplifies this metric is given later in Eq.~\eqref{eq:hopfcoordinates}, via Euler angles; however, such a change complicates the expression of the vector field, so we adhere to Cartesian coordinates in this section.

We now wish to characterize the metric~\eqref{eq:LorentzianmetricCP2}---for example, to determine whether it admits Killing vectors. Explicit computations can be challenging, so we first prove the following theorem:

\begin{theorem}
    Let
    \[
        g^L_{\mu\nu} = g^R_{\mu\nu} - 2 V_\mu V_\nu,
    \]
    where $g^R_{\mu\nu}$ is a Riemannian metric and $V_{\mu}$ is a nowhere-vanishing line field obeying $g_R^{\mu\nu} V_\mu V_\nu = 1$. Let $\xi^\mu$ be a vector field such that $\mathcal{L}_{\xi}V_\mu = \alpha V_\mu$, for some smooth function $\alpha$. Then $\mathcal{L}_{\xi}g^{L}_{\mu\nu} = 0$ iff
    \[
        \mathcal{L}_{\xi}g^{R}_{\mu\nu} = 0 \quad \text{and} \quad \mathcal{L}_{\xi}V_{\mu} = 0.
    \]
    \label{th:killingvectors}
\end{theorem}

\begin{proof}
    A vector field $\xi^\mu$ is Killing for $g^L_{\mu\nu}$ iff
    \[
       \mathcal{L}_{\xi}g^{L}_{\mu\nu} =  \mathcal{L}_{\xi}g^{R}_{\mu\nu} - 2 \mathcal{L}_{\xi}(V_\mu V_\nu).
    \]
    If $\mathcal{L}_{\xi}g^{R}_{\mu\nu} = 0$ and $\mathcal{L}_{\xi}V_{\mu} = 0$, the Leibniz rule gives $\mathcal{L}_{\xi}g^{L}_{\mu\nu}= 0$.

    Conversely, assume $\mathcal{L}_{\xi}g^{L}_{\mu\nu} = 0$. Then
    \[
        \mathcal{L}_{\xi}g^{R}_{\mu\nu} = 2 \left[(\mathcal{L}_{\xi}V_{\mu})V_\nu + V_\mu(\mathcal{L}_{\xi}V_\nu) \right] = 4\alpha V_\mu V_\nu.
    \]

Taking the Lie derivative of the normalization condition:
\begin{align*}
0=\mathcal L_\xi(g_R^{\mu\nu}V_\mu V_\nu)
&=(\mathcal L_\xi g_R^{\mu\nu})V_\mu V_\nu
+2g_R^{\mu\nu}(\mathcal L_\xi V_\mu)V_\nu\\
&=-4\alpha + 2\alpha = -2\alpha
\end{align*}

This implies $\alpha = 0$ and so $\mathcal{L}_{\xi}V_{\mu} = 0$ and $\mathcal{L}_{\xi}g^{R}_{\mu\nu} = 0$.
\end{proof}

By Theorem~\ref{th:killingvectors}, the Lorentzian metric~\eqref{eq:LorentzianmetricCP2} has no obvious Killing vectors, since the Lie bracket of $w^{\mu}$ with each Killing vector of the Fubini–Study metric is non-zero. Indeed, it is possible to prove that \eqref{eq:LorentzianmetricCP2} has no Killing vectors\footnote{In the $U_3$ chart, if there exists $\xi^\mu$ Killing, then $\mathcal{L}_{\xi}R_{\mu\nu\rho\sigma} = 0$ implies, at the origin $\bm{0}$, $\left.\xi^\mu\right|_{\bm{0}} = 0$ and $\left.(\nabla_\mu\xi_\nu)\right|_{\bm{0}} = 0$. It follows then that $\xi^\mu$ is trivial \cite[p. 442]{wald1984general}}.

The absence of symmetries makes the analysis of this spacetime challenging; accordingly, we focus solely on the behavior of timelike geodesic congruences. Specifically, we choose a particular family of geodesic congruences, which we integrate numerically. The initial conditions for the congruence are 11 random points on the unit $3$-sphere centered at the origin of the $U_1$ chart, so that
\[
    x(0)^2+y(0)^2+z(0)^2+t(0)^2 = 1.
\]
The initial derivatives with respect to the affine parameter~$\lambda$ are taken to be the values of the vector field~\eqref{eq:vectorfieldprojective}, normalized with respect to the Fubini-Study metric (with $\Lambda = 6$) at those points. This normalization ensures that the resulting geodesics are timelike; their trajectories are displayed in Fig.~\ref{fig:figure11}.

\begin{figure}
    \begin{subfigure}[b]{0.45\linewidth}
        \centering
        \includegraphics[width=\linewidth]{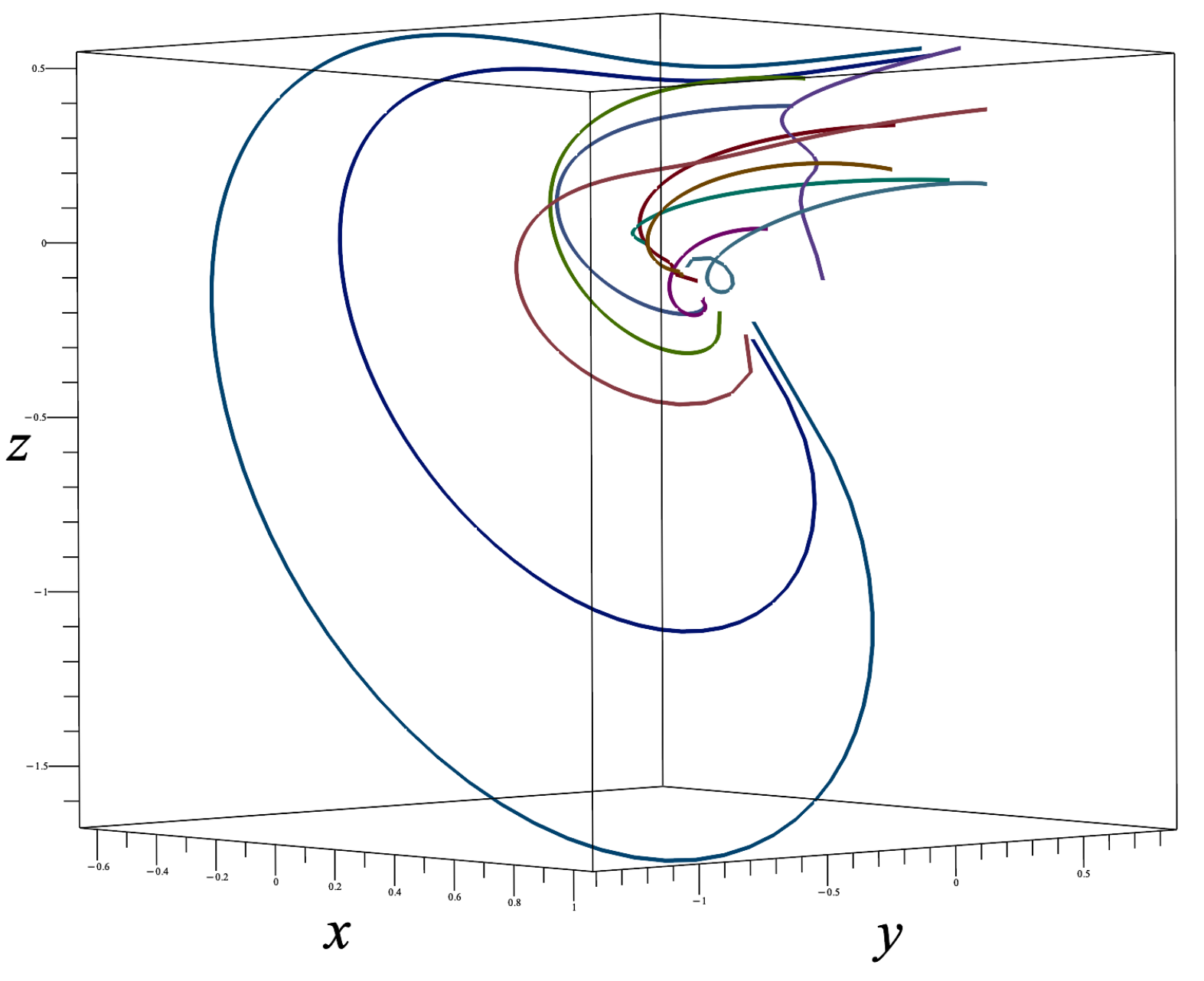}
        \caption{$xyz$ projection}
    \end{subfigure}
    \hfill
    \begin{subfigure}{0.45\linewidth}
        \centering
        \includegraphics[width=\linewidth]{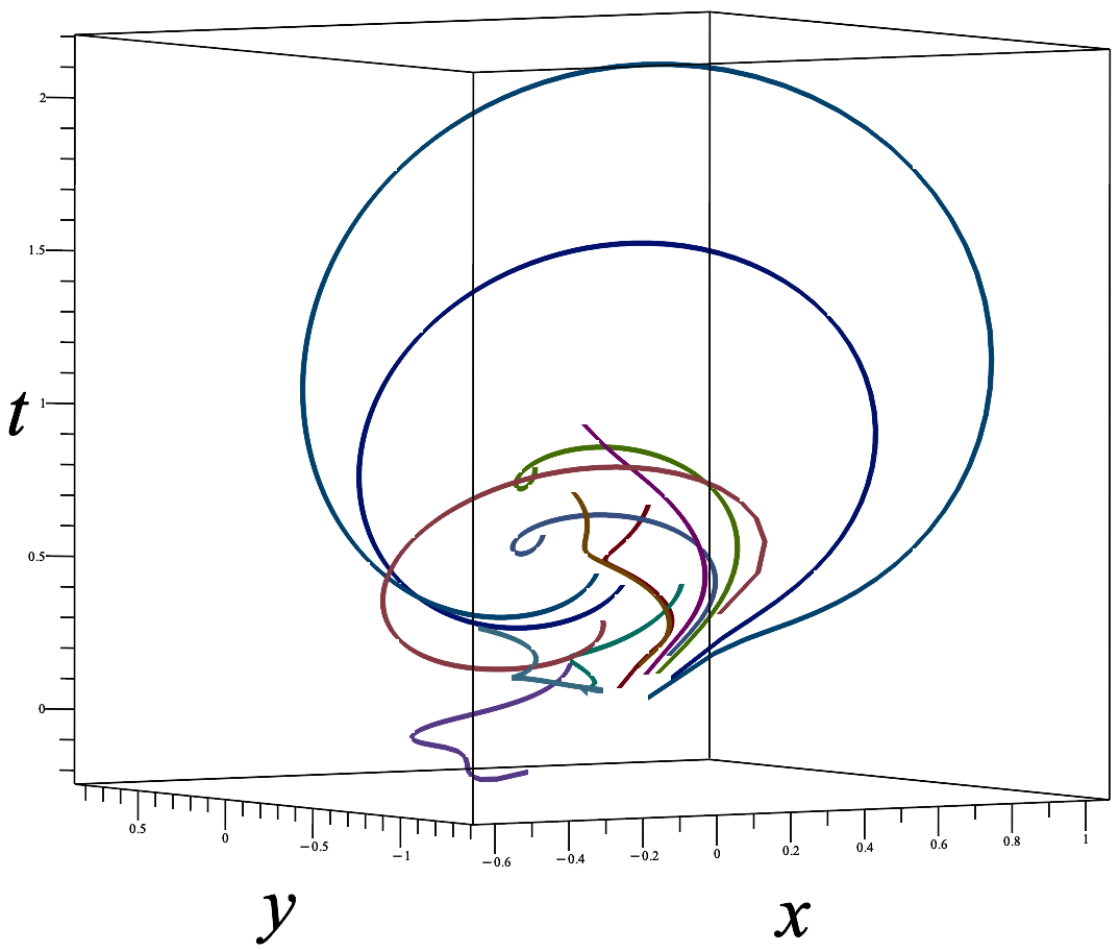}
        \caption{$xyt$ projection}
    \end{subfigure}
        \caption{Numerical solution of the geodesic equations for the metric~\eqref{eq:LorentzianmetricCP2} obtained via a Fehlberg fourth-fifth-order Runge-Kutta scheme. Initial conditions were chosen as 11 random points on the unit $3$-sphere, with the initial velocity given by the (normalized) vector field~\eqref{eq:vectorfieldprojective}. (a): Projection of the trajectories onto the $xyz$ hyperplane. (b): Projection of the trajectories onto the $xyt$ hyperplane.}
    \label{fig:figure11}
\end{figure}

The singular point of the vector field attracts these trajectories, so that the geodesics form circular arcs that terminate at a finite value of the affine parameter at $(0,0,0,0) \in U_1$. However, there exist initial conditions---for instance $x(0) = 1$ with $y(0) = z(0) = t(0) = 0$---for which the timelike geodesic appears to escape to infinity. In such cases it is necessary to switch to other local charts to verify that, within a finite value of the affine parameter, the geodesic returns to the singular point of the vector field.

The behavior of these timelike geodesics in our construction is not surprising. Since the vector field on the closed manifold $\mathbb{CP}^{2}$ possesses a unique singular point, every maximal integral curve originates and terminates at this singularity\footnote{This is not generally true, since closed orbits can exist. However, for the vector field $w^\mu$, there are no closed orbits; they all begin and end at the field’s zero point in the $U_1$ chart.}. These integral curves determine the future orientation of the light-cones for the metric~\eqref{eq:LorentzianmetricCP2}, which explains why the geodesics are drawn toward the singularity. Moreover, apart from the singular point where the aperture vanishes, the light-cones maintain a non-zero opening angle. This permits the formation of CTCs\footnote{An explicit future-directed CTC is given for example by the curve \[\gamma^\mu(s) = \tfrac{1}{100}\Big(2\sin(s), -2 \cos(s),\cos(s),\sin(s)\Big), \ s\in [0,2\pi].\]} and, accordingly, the emergence of chronological horizons. By the argument in Ref.~\cite{chronologyprotection}, these horizons must be compactly generated Cauchy horizons.

We also emphasize that the Morse spacetime $\mathbb{M}$ constructed in the previous section possesses the Morse function as a global time function, implying that $\mathbb{M}$ is stably causal and, by definition, cannot contain any CTCs \cite{Hawking_Ellis_1973}. Therefore, for the change of topology to occur without introducing singularities, any CTCs must be confined to $\mathbb{CP}^{2}$.

We conclude this section with a few remarks on the violation of the standard energy conditions. The WEC requires
\[
    T_{\mu\nu}v^\mu v^\nu \ge 0, \quad \forall\,v^\mu: g^{L}_{\mu\nu}v^\mu v^\nu < 0.
\]
However, considering the vector field $w^\mu$ in~\eqref{eq:vectorfieldprojective}, which is timelike by construction, we find regions where
\[
    T_{\mu\nu}w^\mu w^\nu < 0.
\]
A two-dimensional projection of the region where this inequality holds (evaluated at $x=t=0$) is shown in Fig.~\ref{fig:figure12}.

\begin{figure}
    \centering
    \includegraphics[width=0.8\linewidth]{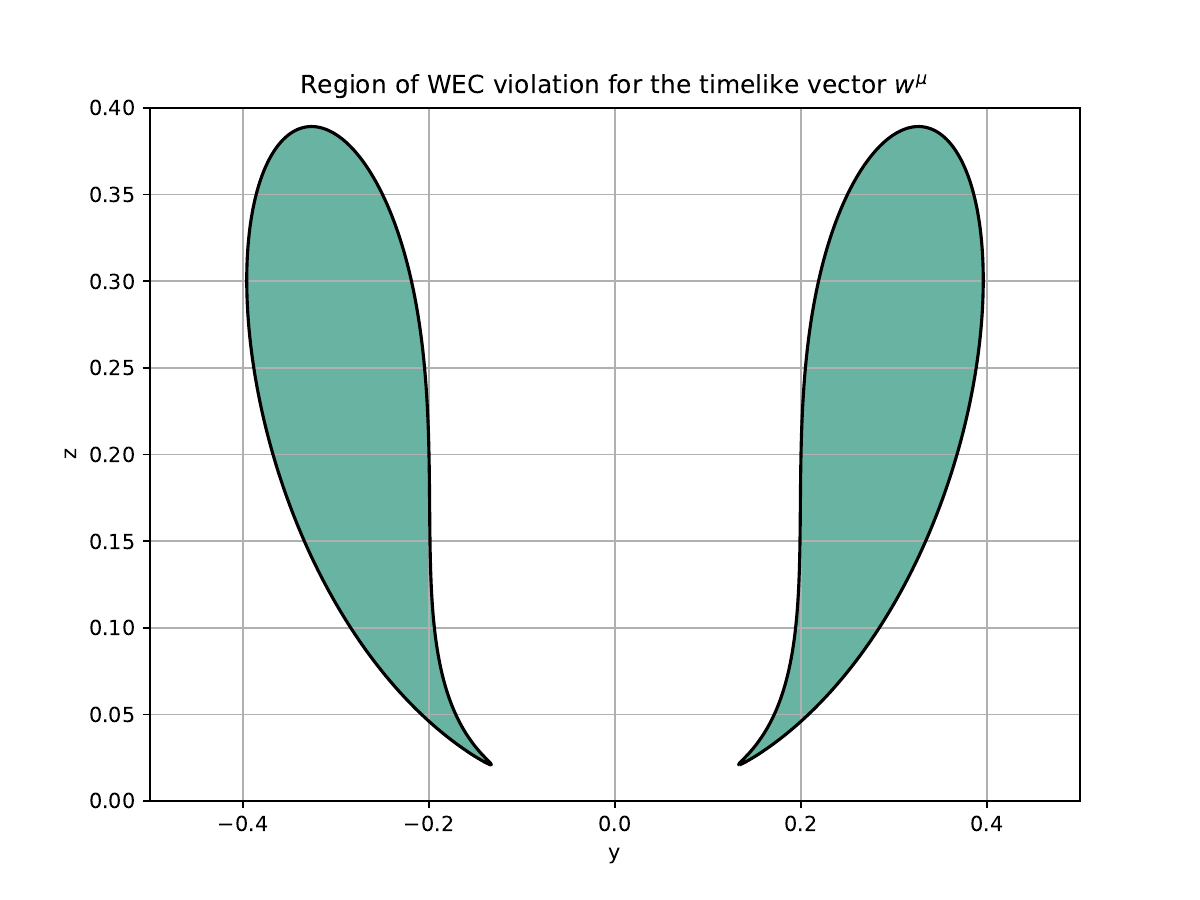}
    \caption{Two-dimensional projection (with $x=t=0$) of the region in which the weak energy condition is violated by the vector field $w^\mu$ in the $U_1$ chart.}
    \label{fig:figure12}
\end{figure}

A more rigorous analysis of the violation of the energy conditions would require a numerical point-wise classification of the Hawking-Ellis types of the stress-energy tensor. This would involve diagonalizing the tensor on a spacetime grid and analyzing its eigenvalues and eigenvectors at each grid point. Here, however, we merely argue for the existence of type-IV regions.

For example, after numerically computing a tetrad $e_{a}^{\mu}$ for the metric~\eqref{eq:LorentzianmetricCP2}, we consider the event $p=(1,0,0,0)$ in the $U_1$ chart. With our specific numerical choice of tetrads, we found that the stress-energy tensor $T^{ab}$ at $p$ has complex eigenvalues, with corresponding complex eigenvectors. So $T^{ab}$ should be of Hawking-Ellis type IV \cite{Hawking_Ellis_1973,Martín–Moruno2017,Maeda_2022}. Because the classification varies continuously, an extended region must exist in which the tensor remains of type IV; in that region all the standard energy conditions are violated.

In the construction~\eqref{eq:lorentzianfromriemannian} the numerical factor $2$ is arbitrary and can be replaced by any real number $\zeta > 1$:
\begin{equation}
   g^{L}_{\mu\nu} = g^{R}_{\mu\nu} - \zeta \, \frac{g^{R}_{\mu\alpha}V^{\alpha}\, g^{R}_{\nu\beta}V^{\beta}}{g^{R}_{\rho\sigma}V^{\rho}V^{\sigma}}.
\end{equation}
The freedom in the choice of the $\zeta$ parameter means that the aperture of the light-cones can be chosen arbitrarily. Consequently, in regions where $T_{\mu\nu}V^{\mu}V^{\nu} > 0$ one can enforce the WEC---and hence, by continuity, the NEC---by taking $\zeta$ sufficiently close to $1$. This is the approach adopted in Ref.~\cite{Yodzis1973}.

In our case, however, changing $\zeta$ merely shifts, thins, or expands the regions where the energy conditions are violated; it does not eliminate them altogether. Nonetheless, we expect that, with an appropriate choice of $\zeta$ and a suitable starting Riemannian metric, one might follow the procedure of Ref.~\cite{Yodzis1973} to remove such regions completely---at least for the Morse spacetime $\mathbb{M}$.

\section{The gluing process}
\label{sec:gluing}

\subsection{The setting}

In this section, we address the gluing process that joins the two singular spacetimes $\mathbb{M}$ and $\mathbb{CP}^{2}$, for which we have so far provided only a topological description at the beginning of Section~\ref{sec:wormholenucleation}, in order to produce a spacetime without singularities. The gluing process in our context requires a technique different from the standard Darmois–Israel formalism because the submanifold along which the spacetimes are joined changes causal character---from timelike to spacelike---passing through a region where it is null.

Fortunately, a suitable formalism exists for this situation and has been developed in Refs.~\cite{MMars1993, PhysRevD.76.044029}. In what follows, we present only the essential formulas, omitting further details contained in the cited references. The symbols introduced here are independent of those in the preceding sections.

We consider two spacetimes $(M^{+}, g^{+}_{\mu\nu})$ and $(M^{-}, g^{-}_{\mu\nu})$ that we wish to glue along the submanifolds $\Sigma^{+} \subset M^{+}$ and $\Sigma^{-} \subset M^{-}$. These submanifolds may change causal character. After the identification $\Sigma^{+} = \Sigma^{-} \equiv \Sigma$, one can introduce intrinsic coordinates $\xi^{a}$ on $\Sigma$. The embeddings of $\Sigma$ into $M^{+}$ and $M^{-}$ are given by
\[
    x^{\mu}_{+} = x^{\mu}_{+}(\xi), \qquad  x^{\mu}_{-} = x^{\mu}_{-}(\xi),
\]
and a basis of tangent vectors on $\Sigma$ is defined as
\[
    e^{\pm \mu}_{a} \equiv \frac{\partial x^{\mu}_{\pm}(\xi)}{\partial \xi^{a}}.
\]

The preliminary junction conditions require that the induced metrics on $\Sigma^{+}$ and $\Sigma^{-}$ match:
\begin{equation}
    h^{+}_{ab}(\xi) = h^{-}_{ab}(\xi),
    \label{eq:firstpreliminaryjunction}
\end{equation}
where
\[
    h^{\pm}_{ab}(\xi) \equiv  g^{\pm}_{\mu\nu} \left( x^{\pm}(\xi) \right) \frac{\partial x^{\mu}_{\pm}(\xi)}{\partial \xi^{a}} \frac{\partial x^{\nu}_{\pm}(\xi)}{\partial \xi^{b}}.
\]

In addition, one must introduce \emph{riggings}, i.e., vector fields $l_{+}^{\mu}$ and $l_{-}^{\mu}$ defined on $\Sigma^{+}$ and $\Sigma^{-}$, respectively, that are everywhere transverse to these submanifolds. Given the associated normal forms $N^{+}_{\mu}$ and $N^{-}_{\mu}$, the transversality condition is
\[
    l_{\pm}^{\mu}N^{\pm}_{\mu} \neq 0.
\]

The gluing can occur only if $l_{+}^{\mu}$ is directed outward from $M^{+}$ and $l_{-}^{\mu}$ is directed inward from $M^{-}$ (or vice versa), and if the following additional conditions are satisfied on $\Sigma$:
\begin{subequations}
    \begin{align}
        \left.g^{+}_{\mu \nu} \, l_{+}^\mu \, l_{+}^\nu\right|_{\Sigma} &= \left.g^{-}_{\mu \nu} \, l_{-}^\mu \, l_{-}^\nu\right|_{\Sigma},\\
        \left.g^{+}_{\mu \nu} \, l_{+}^\mu \, e_{a}^{+\nu}\right|_{\Sigma} &= \left.g^{-}_{\mu \nu} \, l_{-}^\mu \, e_{a}^{-\nu}\right|_{\Sigma}.
    \end{align}
    \label{eq:secondpreliminaryjunction}
\end{subequations}

If~\eqref{eq:firstpreliminaryjunction} and~\eqref{eq:secondpreliminaryjunction} are satisfied\footnote{Even if the induced metrics match on $\Sigma$, the requirements~\eqref{eq:secondpreliminaryjunction} may fail, as happens for two identical copies of the same spacetime with boundary \cite{PhysRevD.76.044029}.}, after the identifications $e^{+\mu}_a = e^{-\mu}_a = e^{\mu}_a$ and $l_{+}^{\mu} = l_{-}^{\mu} = l^{\mu}$, one can avoid the appearance of thin shells or distributional terms in the equations of motion if
\begin{equation}
    \left[\mathcal{H}_{\mu\nu} \right] = 0,
    \label{eq:jumpintrinsiccurvature}
\end{equation}
where the jump $[f]$ is defined as
\[
    [f](p) \equiv \lim_{x\rightarrow p^+} f^{+}(x) - \lim_{x\rightarrow p^-} f^{-}(x) \qquad \forall\, p \in \Sigma,
\]
with the two limits taken in $M^+$ and $M^-$, respectively, and
\[
    \mathcal{H}^{\pm}_{\mu\nu} \equiv \Pi^{\alpha}{}_{\mu}\Pi^{\beta}{}_{\nu} \nabla^{\pm}_{\alpha} l_{\beta}\Big|_{\Sigma}.
\]
Here $\Pi^{\alpha}{}_{\mu}$ is a generalized projector,
\[
    \Pi^{\alpha}{}_{\mu} = \delta^{\alpha}_{\mu} - \frac{1}{n}\, N_{\mu}\, l^{\alpha},
\]
where the identification $N^{+}_{\mu} = N^{-}_{\mu}$ has been made on $\Sigma$, and $n \equiv l_{+}^{\mu}N^{+}_{\mu} = l_{-}^{\mu}N^{-}_{\mu}$.

In the following subsections, we construct the connected sum $\mathbb{M} \# {\mathbb{CP}^{2}}$ by first gluing $\mathbb{CP}^{2}$ to an annulus $S^3_{I} \times I$, where $S^3_{I}$ is a one-parameter family of squashed 3-spheres whose shape varies along $I$, and then gluing a second annulus $S^3_{I} \times I$ to $\mathbb{M}$, thereby demonstrating that no thin shells appear. The two annuli are subsequently attached via an orientation-reversing diffeomorphism, forming the \emph{neck} $S^3 \times I$---not to be confused with the wormhole throat---in which $S^3$ denotes the round 3-sphere. This construction, sketched in Fig.~\ref{fig:figure13}, allows us to express explicitly the homotopy between the vector fields on $\mathbb{M}$ and $\mathbb{CP}^{2}$.

\begin{figure}
    \centering
    \includegraphics[width=0.9\linewidth]{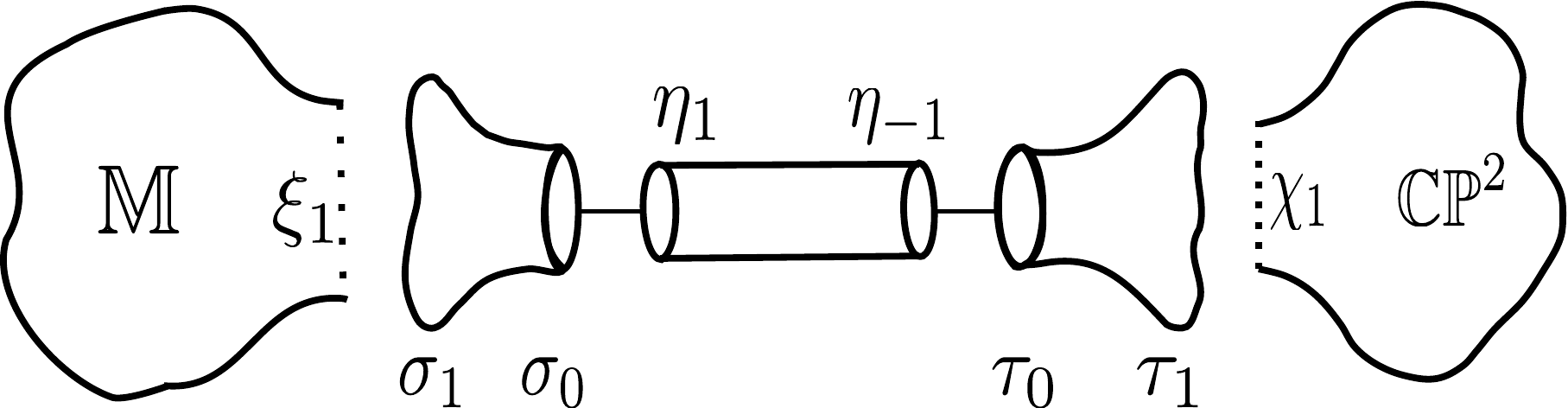}
    \caption{The two annuli are glued, respectively, to $\mathbb{M}$ and $\mathbb{CP}^{2}$ along the surfaces $\xi=\xi_{1}$ and $\chi=\chi_{1}$, each homeomorphic to $S^{3}$. Their purpose is to deform these $3$-spheres continuously into round $3$-spheres, thereby permitting an explicit homotopy between the two vector fields---one with its radial component reversed because of the orientation-reversing diffeomorphism.}
    \label{fig:figure13}
\end{figure}

\subsection{Interface between the complex projective plane and the annulus}
\label{subsec:firstannulus}

In $\mathbb{CP}^{2}$ we consider an embedding of $S^3$ that partitions the manifold into two regions, $\mathbb{CP}^{2}_{+}$ (the exterior) and $\mathbb{CP}^{2}_{-}$ (the interior). We work in the $U_1$ chart, within which we introduce the following coordinate transformation, where $\Lambda > 0$ is a constant:
\begin{equation}
    \begin{aligned}
        x &= \sqrt{\tfrac{6}{\Lambda}} \tan\!\left( \chi \sqrt{\tfrac{\Lambda}{6}}\right) \cos\!\left(\tfrac{\theta}{2}\right) \cos\!\left(\tfrac{\psi}{2}+\tfrac{\phi}{2}\right),\\[1ex]
        y &= \sqrt{\tfrac{6}{\Lambda}} \tan\!\left( \chi \sqrt{\tfrac{\Lambda}{6}}\right) \cos\!\left(\tfrac{\theta}{2}\right) \sin\!\left(\tfrac{\psi}{2}+\tfrac{\phi}{2}\right),\\[1ex]
        z &= \sqrt{\tfrac{6}{\Lambda}} \tan\!\left( \chi \sqrt{\tfrac{\Lambda}{6}}\right) \sin\!\left(\tfrac{\theta}{2}\right) \cos\!\left(\tfrac{\psi}{2}-\tfrac{\phi}{2}\right),\\[1ex]
        t &= \sqrt{\tfrac{6}{\Lambda}} \tan\!\left( \chi \sqrt{\tfrac{\Lambda}{6}}\right) \sin\!\left(\tfrac{\theta}{2}\right) \sin\!\left(\tfrac{\psi}{2}-\tfrac{\phi}{2}\right),
    \end{aligned}
    \label{eq:hopfcoordinates}
\end{equation}
given in terms of the Euler angles and a radial coordinate, with the ranges
\begin{equation}
    \begin{aligned}
        &0 \le \theta \le \pi,\qquad
        0 \le \phi \le 2\pi,\\[1ex]
        &0 \le \psi \le 4\pi,\qquad
        0 \le \chi < \tfrac{\pi}{2}\sqrt{\tfrac{6}{\Lambda}}.
    \end{aligned}
    \label{eq:hopfcoordinatesrange}
\end{equation}
These coordinates cover the entire $U_1$ chart, which is diffeomorphic to $\mathbb{R}^4$, except at the coordinate singularities $\chi=0$ and $\theta \in \{0, \pi\}$. The Fubini-Study metric~\eqref{eq:FubiniStudy} then becomes
\begin{equation}
    \d s^2_{\text{FS}} = \d \chi^2 +A^2(\chi) \left( \bm{\sigma}_1^2 + \bm{\sigma}_2^2+  B^2(\chi)~\bm{\sigma}_3^2 \right),
    \label{eq:fubinistudysu2forms}
\end{equation}
where
\[
    \begin{aligned}
        A^2(\chi) &= \frac{1}{4} \left(\frac{6}{\Lambda} \right)\sin^2\left(\chi\sqrt{\tfrac{\Lambda}{6}}\right), \\
        B^2(\chi) &= \cos^2\left(\chi\sqrt{\tfrac{\Lambda}{6}}\right),
    \end{aligned}
\]
and
\[
    \begin{aligned}
        \bm{\sigma}_1 &= \cos \psi~\d\theta + \sin \psi \sin \theta~\d \phi, \\
        \bm{\sigma}_2 &= -\sin \psi~\d \theta + \cos \psi \sin \theta~\d \phi, \\
        \bm{\sigma}_3 &= \d \psi + \cos \theta~\d \phi,
    \end{aligned}
\]
are the $SU(2)$ left-invariant 1-forms. The Lorentzian metric~\eqref{eq:LorentzianmetricCP2} can therefore be written in the $U_{1}$ chart as
\[
    \d s^2_{+} = \d s^2_{\mathrm{FS}} - 2\,(k^{+}_\mu \d x^\mu)^2,
\]
where $k^{+}_\mu = k^{+}_\mu(\chi,\psi,\theta,\phi)$ is the line field~\eqref{eq:vectorfieldprojective}, normalized with respect to~\eqref{eq:fubinistudysu2forms}. The surfaces $\chi = \const$ are homeomorphic to $S^3$, so the embedded hypersurface along which we perform the gluing can be selected as
\[
    \Sigma^+: \quad \chi = \chi_1, \qquad 0< \chi_1 < \tfrac{\pi}{2}\sqrt{\tfrac{6}{\Lambda}}.
\]
We then choose intrinsic coordinates on $\Sigma^+$, $x^a = \{\psi,\theta,\phi\}$, so that a basis for the tangent vectors to the hypersurface is
\[
    \bm{e}^+_{\psi} = \frac{\partial}{\partial \psi}, \quad \bm{e}^+_{\theta} = \frac{\partial}{\partial \theta}, \quad \bm{e}^+_{\phi} = \frac{\partial}{\partial \phi}.
\]

Similarly, we consider an embedding of $S^3$ in the annulus $S^3_{I} \times I$, that selects a hypersurface $\Sigma^{-}$ along which we glue $\Sigma^{+}$. To this end we define a coordinate $\tau_0\le \tau \le \tau_1$, with $\tau_0 <\tau_1$ positive constants, on the interval $I$ and adopt as an ansatz for the Riemannian metric a biaxial Bianchi-IX metric
\begin{equation}
    \d s^2_{\mathrm{N}} = \d \tau^2 + a(\tau)^2 \, \left(\bm{\sigma}_1^2 +\bm{\sigma}_2^2 + (1-b^2(\tau))~\bm{\sigma}_3^2 \right).
    \label{eq:riemannianmetricneck}
\end{equation}
The corresponding Lorentzian metric uses the same vector field~\eqref{eq:vectorfieldprojective}:
\[
    \d s^2_{-} = \d s^2_{\mathrm{N}} - 2\, (k^{-}_\mu \d x^\mu)^2,
\]
where the line field $k^{-}_\mu = k^{-}_\mu(\tau, \psi, \theta, \phi)$ is defined on $S^3_I \times I$, with the radial $\chi$-coordinate replaced by $\tau$, and normalized with respect to~\eqref{eq:riemannianmetricneck}. We choose the gluing hypersurface to be the boundary $S^3_{\tau_1} \times \{ \tau_1 \}$:
\[
    \Sigma^{-}: \quad \tau = \tau_1.
\]

The functions $a(\tau)$ and $b(\tau)$ describe the embeddings of the 3-spheres in $S^{3}_{I}\times I$ and are constrained by the gluing conditions. In particular, $b(\tau)$ can be chosen to pass smoothly and monotonically from a constant value $b(\tau_1) = \const$ to zero at  $\tau_0$, with vanishing derivative, $b(\tau_0)=\dot{b}(\tau_0) = 0$. This allows the embedded 3-sphere to be deformed continuously from a squashed 3-sphere at $\tau = \tau_1$ to the round 3-sphere at $\tau = \tau_0$.

For the intrinsic coordinates on $\Sigma^{-}$ we adopt the same set used for $\Sigma^{+}$; hence, from now on, we denote $\Sigma^{+} = \Sigma^{-} \equiv \Sigma$ and $e^{+\mu}_{a} = e^{-\mu}_{a} \equiv e^{\mu}_{a}$.

The preliminary junction conditions~\eqref{eq:firstpreliminaryjunction} imply---given that the line fields $k^{\pm}_{\mu}$ coincide on $\Sigma$---that $\chi_1 = \tau_1$ and, up to a sign,
\[
    a(\tau_1) = \frac{1}{2}\sqrt{\frac{6}{\Lambda}} \sin\left( \chi_1 \sqrt{\tfrac{\Lambda}{6}} \right), \qquad
    b(\tau_1) = \sin\left( \chi_1 \sqrt{\tfrac{\Lambda}{6}} \right).
\]
The normal 1-forms on $\Sigma^{+}$ and $\Sigma^{-}$ are chosen as
\[
    N^{+}_{\mu} = \d \chi, \quad N^{-}_{\mu} = \d \tau,
\]
and are identified on $\Sigma$. Their norms change causal character:
\[
    N^{+\mu}N^{+}_{\mu} = 1-2\left(k^{\chi}\big|_{\Sigma}\right)^2, \qquad N^{-\mu}N^{-}_{\mu} = 1-2\left(k^{\tau}\big|_{\Sigma}\right)^2,
\]
passing from timelike to spacelike through a subset $\Sigma_0 \subset \Sigma$ where they are null. Transverse riggings can be picked as
\[
    l^{+\mu} = \frac{\partial}{\partial \chi}, \quad l^{-\mu} = \frac{\partial}{\partial \tau},
\]
and are likewise identified on $\Sigma$. They satisfy $l^{\mu}N_{\mu} \ne 0$ everywhere on $\Sigma$ and, since $l^{+\mu}$ is directed into $\mathbb{CP}^{2}_{+}$ while $l^{-\mu}$ is directed out of $S^3_I \times I$, they also fulfill the second preliminary gluing condition~\eqref{eq:secondpreliminaryjunction}.

The last step is to impose a sufficient condition for the absence of thin shells or distributional terms in the equations of motion. We observe that if the derivatives of the metric are continuous across $\Sigma$,
\begin{equation}
    \left[g_{\mu\nu , \alpha} \right] = 0,
    \label{eq:jumpderivatives}
\end{equation}
then $\left[\Gamma^{\alpha}_{\mu\nu} \right] = 0$ and condition~\eqref{eq:jumpintrinsiccurvature} is satisfied. A sufficient set of matching conditions that ensure the validity of Equation~\eqref{eq:jumpderivatives} are two separate conditions on the jumps of the derivatives of both the Riemannian metric and the line field:
\[
    \left[ g^{R}_{\mu\nu , \alpha} \right] = 0, \qquad
    \left[ k_{\mu, \alpha} \right] = 0.
\]

Because any discontinuity can arise only in the derivatives with respect to $\chi$ and $\tau$, the requirement of no jumps in the derivatives of the Riemannian metric translates into
\[
    \dot{a}(\tau_1) =  \frac{1}{2}\cos\left(\chi_1\sqrt{\tfrac{\Lambda}{6}}\right), \qquad
    \dot{b}(\tau_1) = \sqrt{\frac{\Lambda}{6}} \cos\left(\chi_1\sqrt{\tfrac{\Lambda}{6}}\right)
\]
where a dot denotes differentiation with respect to $\tau$. These conditions remove any discontinuity in the derivatives of the line field as well, thanks to the preliminary equality $\chi_1=\tau_1$ and the fact that $k^{+}_\mu$ and $k^{-}_\mu$ agree to first order across $\Sigma$. Consequently, under this stronger matching conditions the metric is $C^1$ across $\Sigma$ and no distributional terms appear in the equations of motion.

\subsection{Interface between the annulus and the Morse spacetime}
\label{subsec:secondannulus}

The gluing of the other annulus $S^3_I \times I$ to the Morse spacetime $\mathbb{M}$ is carried out in an analogous way. The Riemannian metric on $\mathbb{M}$ is the one induced on the level sets of the Morse function from $\mathbb{R}^5$ in Eq.~\eqref{eq:riemannianstart}, which in the coordinates
\begin{equation}
    \begin{aligned}
        x_1 &= \xi \cos(\psi), \\
        x_2 &= \xi \sin(\psi)\cos(\theta), \\
        x_3 &= \xi \sin(\psi)\sin(\theta) \cos(\phi), \\
        x_4 &= \xi \sin(\psi)\sin(\theta) \sin(\phi),
    \end{aligned} \quad \quad
    \begin{aligned}
        &0 \le \psi \le \pi, \\
        &0 \le \theta \le \pi, \\
        &0 \le \phi \le 2\pi, \\
        &0 < \xi \le + \infty.
    \end{aligned}
    \label{eq:3-sphericalcoordinates}
\end{equation}
takes the form
\begin{equation}
    \d s^{2}_{\M} = \d s^2_{\mathrm{Flat}} + \xi^2 \left(\cos(2 \psi)\,\d \xi - \xi\sin(2 \psi)\,\d \psi \right)^2,
    \label{eq:riemannianmetricmorse}
\end{equation}
where $\d s^2_{\mathrm{Flat}}$ is the flat Euclidean metric
\[
    \d s^2_{\mathrm{Flat}} = \d \xi^2 + \xi^2 \left(\d \psi^2 + \sin^2(\psi)\,\d \Omega_2^2 \right),
\]
with $\d \Omega_2^2$ the line element on the unit 2-sphere. The Lorentzian metric~\eqref{eq:firstmorselorentzian} is then given in compact form as
\begin{equation}
    \d s^2_{+} = \d s^2_{\M} -2 (v^{+}_{\mu} \d x^{\mu})^2,
\end{equation}
where $v^{+}_{\mu} = v^{+}_{\mu}(\xi, \psi, \theta, \phi)$ is the Morse line field, already normalized with respect to~\eqref{eq:riemannianmetricmorse}. The hypersurfaces $\xi = \const$ are inhomogeneous deformations of $S^3$, and we choose
\[
    \Sigma^{+}: \quad \xi = \xi_1, \qquad \xi_1 > 0,
\]
with a basis of tangent vectors $\{\bm{e}^+_{\psi}, \bm{e}^+_{\theta}, \bm{e}^+_{\phi} \}$. As in the previous gluing, on $S^3_I \times I$, we choose the Riemannian metric ansatz
\begin{align*}
    \d s^2_{\mathrm{N}} = \d \sigma^2 + \alpha^2(\sigma) \big[\d \psi^2 + \sin^2(\psi)\,\d \Omega_2^2\big] \\
    +\beta^2(\sigma) \left(\cos(2 \psi)\,\d \sigma - \sigma\sin(2 \psi)\,\d \psi \right)^2,
\end{align*}
where the coordinate on $I$ is such that $\sigma_0 \le \sigma \le \sigma_1$, with positive constants $\sigma_0 < \sigma_1$; the corresponding Lorentzian metric can be written as
\begin{equation}
    \d s^2_{-} = \d s^2_{\mathrm{N}} -2 (v^{-}_{\mu} \d x^{\mu})^2.
\end{equation}

The hypersurface on which we perform the gluing is the boundary $S^3_{\sigma_1} \times \{ \sigma_1 \}$:
\[
    \Sigma^{-}: \quad \sigma = \sigma_1
\]
on which we choose the basis $\{\bm{e}^-_{\psi}, \bm{e}^-_{\theta}, \bm{e}^-_{\phi} \}$. Here the functions $\alpha(\sigma)$ and $\beta(\sigma)$ serve the same role as $a(\tau)$ and $b(\tau)$ in the previous subsection; in particular, $\beta(\sigma)$ can be chosen to vary smoothly and monotonically from a constant value $\beta(\sigma_1) = \const$ to zero at $\sigma_0$, with vanishing derivative, $\beta(\sigma_0) = \dot{\beta}(\sigma_0) = 0$.

The preliminary junction conditions~\eqref{eq:firstpreliminaryjunction} then require that $\xi_1 = \sigma_1$ and, up to a sign,
\begin{align}
    \alpha(\sigma_1) = \xi_1, \qquad \beta(\sigma_1) = \xi_1.
\end{align}
Similarly, the normal 1-forms on $\Sigma^{+}$ and $\Sigma^{-}$ are
\[
    N^{+}_{\mu} = \d \xi, \qquad N^{-}_{\mu} = \d \sigma,
\]
and are identified on $\Sigma$. Transverse riggings are
\[
    l^{+\mu} = \frac{\partial}{\partial \xi}, \qquad l^{-\mu} = \frac{\partial}{\partial \sigma},
\]
again identified on $\Sigma$. They satisfy $l^{\mu}N_{\mu} \ne 0$ everywhere on $\Sigma$ and have the property that $l^{+\mu}$ is directed into $\mathbb{M}_{+}$, while $l^{-\mu}$ is directed out of $S^3_I \times I$, thereby satisfying the second preliminary gluing condition~\eqref{eq:secondpreliminaryjunction}. Applying the same sufficient $C^{1}$-matching conditions as in the previous subsection leads to
\begin{align}
    \dot{\alpha}(\sigma_1) = 1, \qquad \dot{\beta}(\sigma_1) = 1.
\end{align}

\subsection{Homotopy between the vector fields along the neck}
\label{subsec:homotopy}

After gluing the two annuli $S^{3}_{I}\!\times\! I$ to the spacetimes $\mathbb{M}$ and $\mathbb{CP}^{2}$, respectively, we focus on finding a homotopy between the two line fields $v^{-}_{\mu}$ and $k^{-}_{\mu}$ along the neck $S^{3}\!\times\! I$, where $S^{3}$ is the constant-radius round $3$-sphere. In this section, $\eta_{-1}$, $\eta_{-1/2}$, $\eta_{0}$, $\eta_{1/2}$, and $\eta_{1}$ will be arbitrary points such that
\[
    \eta_{-1} < \eta_{-1/2} < \eta_0 < \eta_{1/2} < \eta_{1}.
\]
We assume that the interval $I$ has coordinate $\eta\in[\eta_{-1},\eta_{1}]$ with $\eta_{-1}<0<\eta_{1}$, and take the Riemannian metric to be
\begin{equation}
    \d s^{2}_{\mathrm N}
    =\d\eta^{2}+R^{2}\left(\d\psi^{2}
    +\sin^{2}\!\psi\,\d\Omega_{2}^{2}\right),
    \label{eq:constantRiemannianmetric}
\end{equation}
where $R>0$ is the constant radius of $S^{3}$ and
\[
    0\le\psi\le\pi,\qquad
    0\le\theta\le\pi,\qquad
    0\le\phi\le2\pi.
\]

This metric can be glued without thin shells to the two annuli provided that the radial profile functions $a(\tau)$ and $\alpha(\sigma)$, which determine the radii of $S^{3}_{\tau_{1}}$ and $S^{3}_{\sigma_{1}}$, have smooth transitions
\begin{equation}
    2a(\tau_{0})=\alpha(\sigma_{0})=R, \qquad
    \dot a(\tau_{0})=\dot\alpha(\sigma_{0})=0.
    \label{eq:a_conditions}
\end{equation}

This assumption simplifies our analysis by eliminating potential complications in the deformation of the vector fields; in particular, because the underlying Riemannian metric is fixed, the conditions \eqref{eq:firstpreliminaryjunction} and \eqref{eq:secondpreliminaryjunction} and $\bigl[g^{R}_{\mu\nu,\alpha}\bigr]=0$ are automatically satisfied. Consequently, the only remaining task is to continuously deform the line field $\bigl.k^{-}_{\mu}\bigr|_{S^{3}}$—with its radial component reversed, since the gluing diffeomorphism is orientation-reversing—into the Morse line field $\bigl.v^{-}_{\mu}\bigr|_{S^{3}}$ while keeping both the fields and their derivatives smooth.

To clarify the construction we consider the problem in $\mathbb{R}^{4}$, equipped with its standard Euclidean metric.  In Cartesian coordinates $\{x,y,z,t\}$ the vector field
\begin{equation}
    \label{eq:vectorfieldprojectiveradial}
    \begin{aligned}
        w_{r}^{\mu}=\, &
         \bigl(-x^{2}-y^{2}+z-\tfrac{2x(ty+xz)}{x^{2}+y^{2}+z^{2}+t^{2}}\bigr)\,\partial_{x}\\
        &+\bigl(t-\tfrac{2y(ty+xz)}{x^{2}+y^{2}+z^{2}+t^{2}}\bigr)\,\partial_{y}\\
        &-\bigl(ty+xz+\tfrac{2z(ty+xz)}{x^{2}+y^{2}+z^{2}+t^{2}}\bigr)\,\partial_{z}\\
        &-\bigl(tx-yz+\tfrac{2t(ty+xz)}{x^{2}+y^{2}+z^{2}+t^{2}}\bigr)\,\partial_{t}
    \end{aligned}
\end{equation}
is obtained by radially reflecting the field~\eqref{eq:vectorfieldprojective}---the subscript $r$ stands for ``radial reflection''. Note that the constraint $x^{2}+y^{2}+z^{2}+t^{2}=R^{2}$ is understood from now on, and $w_{r}^{\mu}$ is not normalized.
By contrast, the Morse vector field is
\[
    v^{\mu}=-x\,\partial_{x}+y\,\partial_{y}+z\,\partial_{z}+t\,\partial_{t}.
\]

Both fields vanish only at the origin, so neither has singularities on $S^{3}$. Normalized with respect to the Euclidean metric they define the maps
\[
  \frac{w_{r}^{\mu}}{\lVert w_{r}\rVert},
  \frac{v^{\mu}}{\lVert v\rVert}: S^{3}\longrightarrow S^{3},
\]
and a degree calculation confirms that each map has degree~$-1$. Therefore, by the Hopf theorem~\cite{Milnor1965} there exists a homotopy
\[
    H(\eta,x^{\mu}): [\eta_{-1},\eta_1]\times S^{3}\longrightarrow S^{3},
\]
with
\[
    H(\eta_{-1},x^{\mu})=\frac{w_{r}^{\mu}}{\lVert w_{r}\rVert},\qquad
    H(\eta_1,x^{\mu})=\frac{v^{\mu}}{\lVert v\rVert},
\]
that is nonsingular for all $\eta\in[\eta_{-1},\eta_1]$ and $x^{\mu}\in S^{3}$. Our goal is to build such a homotopy explicitly in $\mathbb{R}^{4}$ and then interpret it on the neck $S^{3}\times I$.

We first decompose $w_{r}^{\mu}$ as
\begin{equation}
    \label{eq:decomposition1}
    \begin{aligned}
        &w_{r}^{\mu}=\\&\underbrace{-(x^{2}+y^{2}-z)\,\partial_{x}+t\,\partial_{y}-(ty+xz)\,\partial_{z} -(tx-yz)\,\partial_{t}}_{w_{1,r}^{\mu}}\\
                   &\underbrace{-\tfrac{2(ty+xz)}{x^{2}+y^{2}+z^{2}+t^{2}} \bigl(x\,\partial_{x}+y\,\partial_{y}+z\,\partial_{z}+t\,\partial_{t}\bigr)}_{w_{2,r}^{\mu}}.
    \end{aligned}
\end{equation}
Because $w_{1,r}^{\mu}$ shares the same zero and degree as $w_{r}^{\mu}$, we define the first homotopy as
\[
    H_{1}(\eta,x^{\mu})=
    \frac{w_{1,r}^{\mu}+\bigl(1-b(\eta)\bigr)w_{2,r}^{\mu}}
       {\bigl\lVert w_{1,r}^{\mu}+\bigl(1-b(\eta)\bigr)w_{2,r}^{\mu}\bigr\rVert}, \quad \eta \in [\eta_{-1},\eta_{-1/2}],
\]
where $b:[\eta_{-1},\eta_{-1/2}]\to[0,1]$ is a smooth, monotone function with $b(\eta_{-1})=0$ and $b(\eta_{-1/2})=1$. Thus $H_{1}$ deforms $w_{r}^{\mu}$ to $w_{1,r}^{\mu}$ while remaining nonsingular.

Next we introduce
\[
    s^{\mu}\equiv
    (-x^{2}-y^{2}+z)\,\partial_{x}+t\,\partial_{y}-(ty+x)\,\partial_{z}
    -(tx-y)\,\partial_{t},
\]
which again has the same zero and degree as $w_{1,r}^{\mu}$. A second homotopy is
{\small\[
    H_{2}(\eta,x^{\mu})=
    \tfrac
    {(z-x^{2}-y^{2})\partial_{x}+t\partial_{y}-t(y\partial_{z}+x\partial_{t})-(c(z-1)+1)(x\partial_{z}-y\partial_{t})}{\bigl\lVert(z-x^{2}-y^{2})\partial_{x}+t\partial_{y}-t(y\partial_{z}+x\partial_{t})-(c(z-1)+1)(x\partial_{z}-y\partial_{t})\bigr\rVert},
\]}
with $c: [\eta_{-1/2}, \eta_{0}] \to [0,1]$, $c(\eta_{-1/2}) = 1$, and $c(\eta_{0}) = 0$. This deforms $w_{1,r}^{\mu}$ to $s^{\mu}$. Decompose $s^{\mu}$ further as
\[
    s^{\mu}=s_{1}^{\mu}+s_{2}^{\mu},
    \qquad
    s_{1}^{\mu}=z\,\partial_{x}+t\,\partial_{y}-x\,\partial_{z}+y\,\partial_{t}.
\]
Because $s_{1}^{\mu}$ has the desired degree, define
\[
    H_{3}(\eta,x^{\mu})=
    \frac{s_{1}^{\mu}+\bigl(1-b(\eta)\bigr)s_{2}^{\mu}}
       {\bigl\lVert s_{1}^{\mu}+\bigl(1-b(\eta)\bigr)s_{2}^{\mu}\bigr\rVert}, \quad \eta \in [\eta_0,\eta_{1/2}],
\]
so that $H_{3}(\eta_{1/2},x^{\mu})=s_{1}^{\mu}/\lVert s_{1}^{\mu}\rVert$. Finally, $s_{1}^{\mu}$ is related to $v^{\mu}$ by successive rotations in the $xz$, $yt$, and $xy$ planes:
\[
  R = \begin{pmatrix}
        0 & 0 & 1 & 0 \\
        0 & 0 & 0 & 1 \\
        1 & 0 & 0 & 0 \\
        0 & 1 & 0 & 0
    \end{pmatrix}  \in SO(4), \quad Rs_1 = v.
\]
Since $SO(4)$ is connected, we choose a smooth path $R(\alpha(\eta))$ with $R(\alpha(\eta_{1/2})) = \mathbb{I}$ and $R(\alpha(\eta_1))= R$ and set
\[
    H_4(\eta, x^\mu) = R(\alpha(\eta))~\frac{s_1^\mu}
       {\bigl\lVert s_1^\mu \bigr\rVert}, \quad \eta \in [\eta_{1/2},\eta_1].
\]
The functions $b(\eta)$, $c(\eta)$, and $\alpha(\eta)$ are chosen smooth with $\dot b$, $\dot c$, and $\dot\alpha$ vanishing at the matching points between the homotopies. Because the entire deformation is smooth and nonsingular, the resulting Lorentzian metric on the neck remains free of singularities, and no further explicit expression is required for our purposes.

\section{Conclusions}
\label{sec:conclusions}

In this work, we have constructed an explicit Lorentzian cobordism for a topology‐changing spacetime that describes the nucleation of a wormhole. Although topological surgery requires the spacetime to form a singularity at the transition point, by employing the Misner trick and taking the connected sum with $\mathbb{CP}^{2}$ we obtained a non-singular cobordism in which the singularity is replaced by a region containing CTCs.

This connected sum can be thought of as a ``Lorentzian blow-up'', or more precisely an anti-complex blow-up~\cite{scorpan2022wild} of $\mathbb{M}$ at the point $p$ corresponding to the naked singularity. The event $p$ is then replaced by $\mathbb{CP}^1$ and the gradient of the Morse function is smoothly extended over the blow-up.

Indeed, $\mathbb{CP}^{2}\setminus\text{int}(D^4)$ is nothing more than the disk bundle $\mathcal{O}(1) \rightarrow \mathbb{CP}^1$ whose zero-section is a copy of $\mathbb{CP}^1$ on which the fibers of the Hopf projection collapse, and, as we saw in Section~\ref{subsec:complexprojectiveplane}, it is not part of the singular holomorphic foliation---nor could it be, since the no-hair theorem would force singularities if the field were tangent. This $2$-sphere can be viewed as a ``bubble of nothing'' that carries a non-trivial, topologically persistent gravitational charge that cannot be transformed away.

The interpretation of this desingularization can be framed through the expression ``physicist's topology'' in Ref.~\cite{PhysRevD.41.1116}: on large scales (compared to $\Lambda$) our cobordism admits a time function and is therefore stably causal---a property that, in the case of topological transitions, would normally require the formation of a singularity---whereas on small scales one can probe the topology of the spacetime and thus detect the compact $4$-manifold that contains the CTCs.

The resulting picture is that of a wormhole that exists only for times $t>0$ and whose would-be singular neighborhood is replaced by a pocket, namely the $\mathbb{CP}^{2}$: an observer or a light ray that would otherwise end up at the naked singularity enters this pocket and either exits after the wormhole has already formed or follows a closed causal curve.

The same picture applies to the scalar curvature of the initially separate regions of spacetime: rather than diverging until they ``touch'' and create a singularity, these regions are twisted inside $\mathbb{CP}^{2}$ and re-emerge as the wormhole throat. As anticipated, however, this spacetime violates all the standard energy conditions.

Nevertheless, the construction may shed light on the possible relationship between causality violations and the loss of predictability caused by singularities.

Already in Ref.~\cite{PhysRevD.58.023501} it was shown that closed causal curves can replace the Big Bang singularity in cosmological models. Moreover, in Refs.~\cite{IVANENKO1982341, Sardanashvily1985, Rosquist1983, Caustics} it was pointed out that singularities in general relativity can arise---when the Lorentzian metric is written in terms of a Riemannian metric and a vector field---either from the Riemannian metric (meaning that the manifold topology is incompatible with the metric topology) or from the vector field.

In the latter case it may be possible to interpret a class of singularities via their index and thus trace them back to topological transitions. The formation of such singularities would then signal a topology change rather than a failure of the theory, analogous to the singularities that appear in other geometric flows.

In future work, our aim is therefore to characterize all relevant topological surgeries, including the $1$-surgery, in the Lorentzian context and to develop an accompanying physical interpretation. The control provided by topological surgery over singularity formation should clarify whether---and under what circumstances---such metrics can be desingularized via CTCs.

\begin{acknowledgments}
    A. P., S. A., L. H. K., and S. L. thank Alexandros Kehagias for valuable discussions. A. P. also thanks José M.M. Senovilla, Tina Harriott, and Jeff G. Williams for their insightful suggestions. B. S. acknowledges the support of the Natural Sciences and Engineering Research Council of Canada (NSERC), RGPIN-2024-04063.
\end{acknowledgments}

\bibliography{WormholeNucleation}
\end{document}